\documentclass[a4paper,12pt]{amsart}
\usepackage[left=1in, right=1in, top = 1.5in, bottom = 1.5in]{geometry}
\usepackage[round]{natbib}
\usepackage[utf8]{inputenc}
\usepackage[all]{xy}
\usepackage{amssymb}
\usepackage{amsmath}
\usepackage{csquotes}
\usepackage{subfig}
\usepackage{enumerate}
\usepackage{graphicx}
\usepackage{amsaddr}
\usepackage{color}
\usepackage{mwe}
\usepackage{appendix}
\usepackage{centernot}
\usepackage{pgfplots}
\pgfplotsset{compat=1.18}
\usepackage{mathtools}
\usepackage{comment}
\usepackage{makecell}
\usepackage[ruled,vlined]{algorithm2e}
\usepackage{enumitem}

\newcommand{\independent}{\perp\mkern-9.5mu\perp}

\usepackage{tikz}
\usetikzlibrary{arrows.meta,shapes}
\usetikzlibrary{arrows,shapes.arrows,shapes.geometric,shapes.multipart,
decorations.pathmorphing,positioning,shapes.swigs,}

\usepackage{setspace}

\newtheorem{proposition}{Proposition}
\newtheorem{corollary}{Corollary}
\theoremstyle{definition}
\newtheorem{assumption}{Assumption}
\newtheorem{hypothesis}{Hypothesis}

\newtheorem{definition}{Definition}
 
\theoremstyle{remark}
\newtheorem{remark}{Remark}

\title{A simple and powerful test of vaccine waning}

\author[Perényi, Janvin, Stensrud]
{Gellért Perényi\textsuperscript{1},\;
 Matias Janvin\textsuperscript{2,3},\;
 Mats J.\ Stensrud\textsuperscript{1,*}}

\address{
\textsuperscript{1}Institute of Mathematics, Ecole Polytechnique Fédérale de Lausanne, Switzerland ,\\
\textsuperscript{2}Oslo Centre for Biostatistics and Epidemiology, Department of Biostatistics, University of Oslo, Oslo, 0372, Norway,\\
\textsuperscript{3}Department of Surgery and Anesthesiology, Diakonhjemmet Hospital, Oslo, Norway
}
\thanks{\textsuperscript{*}\;Corresponding author: \texttt{mats.stensrud@epfl.ch}}

\begin{document}

\begin{abstract}
    Determining whether vaccine efficacy wanes is important for individual and public decision making. Yet, quantification of waning is a subtle task. The classical approaches cannot be interpreted as measures of declining efficacy unless we impose unreasonable assumptions. Recently, formal causal estimands designed to quantify vaccine waning have been proposed. These estimands can be bounded under weaker assumptions, but the bounds are often too wide to make claims about the presence of waning. We propose a different approach: a formal test to assess whether a treatment effect is constant over time at the individual level. This test provides a considerable power gain over existing approaches and is valid under interpretable assumptions in vaccine trials. We illustrate the increase in power through real and simulated examples, using three different approaches to compute the test statistics. Two of these approaches are based solely on summary data, accessible from existing clinical trials. Beyond our test, we also give new results that bound the waning effect. We use our methods to reanalyze data from a randomized controlled trial of the BNT162b2 COVID-19 vaccine. While prior analysis did not establish waning, our test rejects the null hypothesis of no waning.
\end{abstract}

\maketitle

\vfill

\newpage
\clearpage
    
\section{Introduction}\label{SEC: Intro}
Consider a randomized controlled trial (RCT) evaluating a vaccine against an infectious disease (e.g., HIV). Using data from the RCT, researchers estimated vaccine efficacy (VE) to be $0.8$ one month after vaccination. To assess if the VE changed at later times, the participants who were event-free in the first analysis were included in a subsequent analysis at three months that gave a $\widehat{VE}=0.6$, see Figure \ref{FIG: 2-arms}. The researchers concluded that the efficacy had decreased at three months; that is, the vaccine effect had waned. 

\begin{figure}
    \centering
    \includegraphics[width=0.7\linewidth]{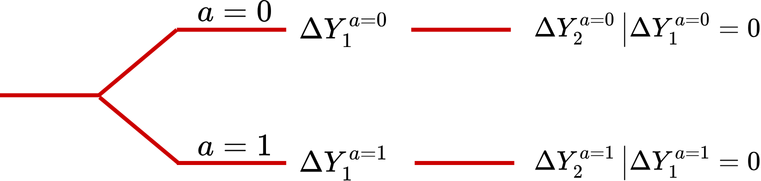}
    \caption{Illustration of outcomes in a 2-arm randomized controlled trial over two intervals. Second interval observation is conditional on being event-free at time 1.}
    \label{FIG: 2-arms}
\end{figure}

A key question is whether such a contrast of VEs at different times gives a meaningful notion of waning. To answer this question, we take for granted that the term waning alludes to a declining causal effect; that is, vaccine protection that decreases over time. To draw valid inference on a causal effect, the populations being compared under vaccination or no vaccination (e.g., placebo) should be exchangeable. Conceptually, this means that the distributions of risk factors for infection are expected to be similar in the two treatment groups  (see for example, \citet{hernan_causal_2023}). Randomization ensures exchangeability for the comparison at the time of the first analysis (at one month). By the time of the subsequent analysis (at three months), however, conditioning on being event-free induces post-randomization selection: the participants who were event-free under vaccination likely differ from those who were event-free without vaccination. There is an additional problem when comparing VEs over time. Even within the same treatment arms, the analyses at the two subsequent intervals are estimated in different populations, making the declining VEs hard to interpret. See Figure \ref{FIG: Invariant dist} for the changing distribution of susceptibility over time and across treatment arms.

These concerns are not merely theoretical; heterogeneity in susceptibility is widespread. To fix ideas, we give an extreme example. Consider a trial of a vaccine against CCR5-tropic HIV \citep{piccininni_immune-selection_2025}. Suppose that all individuals were in contact with the infectious agent, that is, they were naturally ``challenged", within the first month. The VE was estimated to be $0.8$ after one month. Suppose further that the individual-level effect of the vaccine is constant in time: more precisely, had we challenged individuals with an infectious agent immediately after a one-month isolation period, we would have exactly the same outcome as if we had challenged the individuals immediately after a three-month isolation period. Then, the untreated participants who were not infected after the challenge must be inherently immune to the disease (e.g., homozygous for the CCR5 $\Delta32$ gene variant). Thus, by the time of the second analysis, all unprotected trial participants, regardless of the arm, have already experienced the outcome, while those who are protected, either by the vaccine or their genetically induced immunity, will remain event-free at time 2, either by the lack of exposure or by their immunity. It follows that no new events will occur in either of the two arms, and the conventional VE estimate would be equal to $0$. This estimate, however, contradicts the fact that the vaccine effect, on an individual-level, is time-invariant.

This extreme example illustrates a point: conventional VE estimates cannot be used to establish waning. Such analyses are subject to selection bias due to depletion of susceptibles, corresponding to well-known issues with the interpretation of hazard ratios, see, e.g., \citet{hernan_hazards_2010, dumas_how_2025}. In practice, the magnitude of bias is driven by unknown heterogeneity in the outcome risk and can, in principle, be large, although likely not as extreme as in our illustrative example. However, we can quantify the difference between conventional waning estimands and causal waning estimands under given levels of heterogeneity, see Appendix \ref{APP: heterogen} for details.

Recognizing the problems with the conventional waning analyses,  \citet{janvin_quantification_2024} proposed an alternative estimand for VE, based on the idea of challenge trials, building on \citet{stensrud_identification_2023}. In particular, their estimand could be computed in a hypothetical target trial with controlled exposure of the participants \citep{bernstein_norovirus_2015, schmoele-thoma_vaccine_2022}. Such trials are rarely conducted, as exposing participants, regardless of their health status, to infectious pathogens is unethical \citep{hope_challenge_2004}. Vaccinating a (random) subset of participants before pathogen emergence is sometimes a compelling alternative to a challenge trial \citep{lipsitch_depletion--susceptibles_2019, ray_depletion--susceptibles_2020}, but requires knowledge of the timing of the infectious and non-infectious periods. While this design does not permit point identification of the challenge effects without additional assumptions, it can detect waning under mild assumptions. Furthermore, with data from conventional RCTs, the waning estimand of interest is only partially identifiable, i.e., it is bounded under reasonable assumptions. Here, we introduce another approach to assess the existence of vaccine waning. Our work is grounded on a formal causal formulation of a declining treatment effect, related to \citet{janvin_quantification_2024}. However, unlike \citet{janvin_quantification_2024}, we formulate a sharp null hypothesis of no waning and propose a test that is powerful, simple, and assumption-lean. In particular, the test can be conducted under reasonable assumptions in conventional vaccine trials, either using individual-level or summary data. We furthermore suggest a strategy in Appendix \ref{APP_sub: Challenge sens}, based on some additional in-principle testable assumptions, to improve the bounds derived by \citet{janvin_quantification_2024}. 

\section{Observed Data Structure}
Consider i.i.d.\ data from an RCT in which individuals are assigned to treatment $A \in \{0,1\}$ corresponding to placebo or vaccine. Define $Y_k \in \{0,1\}$ as an indicator of at least one outcome occurring by the end of interval $k$. Therefore, $Y_k\geq Y_{k'}$ for all $k>k'$. Then, $\Delta Y_k := Y_k-Y_{k-1}$ is an indicator of whether the first outcome has occurred in interval $k$. To simplify the presentation, we will treat these intervals as discrete. Yet, the test is valid in a continuous-time setting as well, with minor modifications. As the data are assumed to be i.i.d., the subscripts indicating individuals are omitted.
The notation in Figure \ref{FIG: 2-arms} stands for potential outcomes under different levels of treatment. For example, $\Delta Y_1^{a=0}$ is the outcome at time 1 under an intervention that sets the treatment to $a=0$.

As highlighted in Section \ref{SEC: Intro}, the VE at time 2 is not a causal contrast. Consider instead a contrast among those who would survive the first interval, regardless of the treatment assignment,

\begin{align*}
    &E[\Delta Y_2^{a=1}| \Delta Y_1^{a=1}=\Delta Y_1^{a=0}=0] \quad \text{vs} \\
    &E[\Delta Y_2^{a=0}| \Delta Y_1^{a=1}=\Delta Y_1^{a=0}=0].
\end{align*}
By restricting the analysis to those who survive under both treatment assignments, the two populations are comparable at time 2. This contrast is an example of principal stratum effects (PSEs) \citep{frangakis_principal_2002}.

However, we need to be careful when comparing effect contrasts among those who would survive under both levels of treatment at the second time to the same contrast among the general population at the first time: even if the analysis at each time has a causal interpretation, the two comparisons are computed in different populations. This complication concerning the PSE is particular to our setting and goes beyond the well-known critique of the PSE being a cross-world estimand \citep{robins_new_1986, richardson_single_2013, stensrud_conditional_2023}.  

To illustrate our target estimand, consider instead an unconventional randomized 4-arm trial where participants are randomly assigned to one of the two levels of treatment and one of the two possible times of exposure (visualized in Figure \ref{FIG: 4-arms}):

Define exposure at time $k$ as $E_k \in \{0,1\}$ for $k\in \{1,2\}$, corresponding to isolation ($E_k=0$), or exposure through some controlled procedure ($E_k=1$). We will not presume that we have data from this unconventional trial, but instead use this as a target trial, motivating our estimand. 

In our derivations, we will use the following assumption from \citet{stensrud_identification_2023}:
\begin{assumption}[Exposure necessity]\label{ASS: Exp nec}
    $$
    E_k^a =0 \implies  \Delta Y_k^{a}=0; \; E_2^{a, e_1=0}=0 \implies \Delta Y_2^{a, e_1=0} =0 ,
    $$
    for all $a \in \{0,1\}, k \in \{1,2\}.$
\end{assumption}
The intervention on exposure happens at two times in the 4-arm trial. The potential outcome at time 1 is defined under the intervention setting exposure to 1, and the potential outcome at time 2 is defined under an intervention setting exposure at time 1 to be 0 and exposure at time 2 to be 1. Thus, in this hypothetical trial, participants are exposed to the pathogen according to a controlled procedure at a specific time.

The potential outcome at time 2 is defined under an intervention that isolates individuals at time 1. Thus,  those who were assigned to $e_1=0, e_2=1$ will not develop the outcome at the first time; they cannot be infected without exposure, following Assumption \ref{ASS: Exp nec}. The random assignment of treatment and exposure implies that the populations in all four arms are exchangeable, see Figure \ref{FIG: Invariant dist} for an illustration.

\begin{figure}
    \centering
    \includegraphics[width=0.7\linewidth]{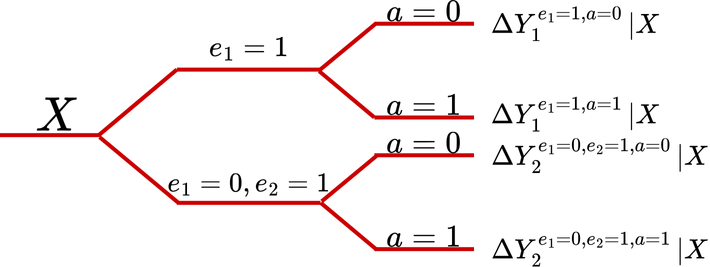}
    \caption{4-arm randomized controlled trial: Individuals are randomized to treatment and to the time of the exposure. Due to isolation, those who are assigned to time 2 exposure are exchangeable with those who are assigned to exposure at time 1. $X$ denotes the distribution of the baseline characteristics in the trial population, which is, by randomization and Assumption \ref{ASS: Exp nec}, expected to be identical in all 4 arms.}
    \label{FIG: 4-arms}
\end{figure}

As proposed by \citet{janvin_quantification_2024}, we define the challenge effects as follows:

\begin{definition}[Marginal challenge effect]\label{DEF: Chall Eff}
    \begin{align}
        VE_1^{\mathrm{challenge}} &= 1- \frac{E[\Delta Y_1^{a=1, e_1=1}]}{E[\Delta Y_1^{a=0, e_1=1}]},\label{EQ: VE_ch_1}\\
        VE_2^{\mathrm{challenge}} &= 1- \frac{E[\Delta Y_{2}^{a=1, e_1=0, e_2=1}]}{E[\Delta Y_{2}^{a=0, e_1=0, e_2=1}]} \label{EQ: VE_ch_2}
    \end{align}
\end{definition}
The challenge effect conditional on baseline covariates is defined analogously, see Appendix \ref{APP:Cond wane} for details.

The challenge effect quantifies the effect of the vaccine in each interval in a way that is insensitive to the observed exposure pattern and the changing distribution of the unobserved heterogeneities over time. If $VE_1^{\mathrm{challenge}}$ exceeds $VE_2^{\mathrm{challenge}}$, then the protective effect of the vaccine has decreased over time, within the population that has been defined pre-randomization 
\citep{janvin_quantification_2024}. 

While a ``challenge trial", such as the one illustrated in Figure \ref{FIG: 4-arms}, is sufficient to identify both $VE_1^{\mathrm{challenge}}$ and $VE_2^{\mathrm{challenge}}$, conducting such a trial is usually unethical, as deliberately exposing individuals to the pathogen could cause  severe harm; hence, they are rarely conducted. Challenge trials can be run in healthy volunteers to minimize the risk of severe adverse effects; however, this approach typically fails to reflect the population of interest.

\citet{janvin_quantification_2024} derived conditions under which a waning parameter can be partially identified, even if the data are from a standard 2-arm RCT, in which exposure was neither intervened on nor measured, such as the one illustrated in Figure \ref{FIG: 2-arms}. Yet, the bounds derived for the second interval challenge effect can be wide, often including the first interval challenge effect, which prohibits a firm claim about the waning of the vaccine.

\begin{figure}
    \centering
    \includegraphics[width=1\linewidth]{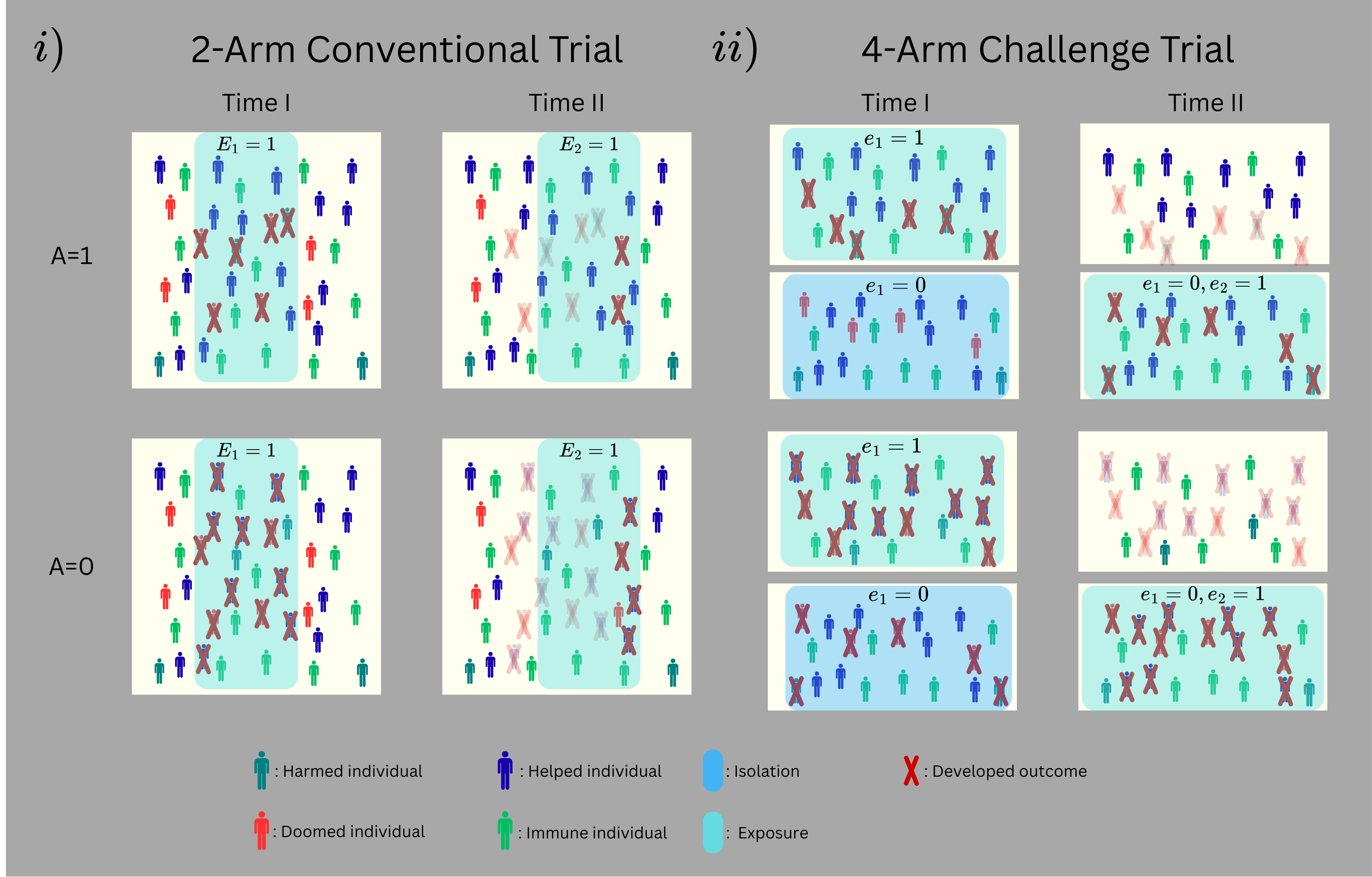}
    \caption{Illustration of two different trial designs. In conventional trials, Subfigure i), the populations under treatment and no treatment are exchangeable at baseline, without intervening on the exposure. Among exposed at time 1, the treated and the untreated groups are exchangeable as well, assuming treatment blinding, yielding an observed vaccine efficacy (VE) of $0.5$ in this example. However, by time 2, the depletion of susceptibles (i.e., selective exclusion from the at-risk population of the doomed and harmed or doomed and helped strata, depending on the assigned treatment), alters the distribution of susceptibility of the at-risk population. Even without waning, as the number of susceptibles declines more rapidly among the untreated, the observed vaccine efficacy declines to approximately $0.41$. In contrast, in a challenge trial, Subfigure ii), the four arms are exchangeable at baseline. Isolation ensures participants assigned to exposure at time 2 remain at risk until then, making them exchangeable at the time of the exposure as well. Consequently, the challenge vaccine efficacy is $0.5$ at both times, equal to the proportion in the doomed-or-harmed populations divided by the proportion in the doomed-or-helped populations in this hypothetical example.}
    \label{FIG: Invariant dist}
\end{figure}

\section{A formal test}\label{SEC: Inference}

Consider the following sharp null hypothesis of no waning:

\begin{hypothesis}[Sharp null hypothesis]\label{HYP: formal null}

    $$
        \Delta Y_1^{a, e_1=1} = 
        \Delta Y_2^{a, e_1=0, e_2=1} \; w.p.1 \; \forall  a \in \{0,1\}.
    $$
\end{hypothesis}

The sharp null hypothesis (\ref{HYP: formal null}) states that anyone who would experience the outcome under exposure would do so regardless of whether exposure occurs at time 1 or at time 2. In short, individual susceptibility to exposure is constant across intervals and independent of exposure timing. One interpretation of Hypothesis \ref{HYP: formal null} is that the vaccine has an all-or-nothing effect \citep{halloran_thirty-five_2024}; that is, under exposure and a given level of treatment, the potential outcomes are deterministic. This test might reject the null hypothesis of no sharp waning when potential outcomes are sensitive to exposure intensities, which can arise in settings with multiple strains circulating in the population. Thus, we also introduce an alternative test, which allows the exposure to take a continuous value, building on \citet{janvin_quantification_2024}, see  Appendix \ref{APP: multi exp} for details. However, this test often has substantially lower power.

We further restrict the exposure intervention to have no carry-over effect on the outcome. That is, exposure at time 1 cannot affect the outcome at time 2 except through survival. This restriction is compatible with a fixed incubation period of an infectious agent. 
\begin{assumption}[Exposure exclusion restriction]\label{ASS: exp eff restriction}
    $$\Delta Y_2^{a, e_1, \delta y_1=0, e_2=1}=\Delta Y_2^{a, \delta y_1=0, e_2=1} \quad \forall a, e_1 \in \{0,1\}^2$$
\end{assumption}

Assumption \ref{ASS: exp eff restriction} holds if all causal paths from $E_1$ to $\Delta Y_2$ are intersected by $\Delta Y_1$ or $E_2$. However, it fails if exposure followed by survival at time 1 creates sustained immunity that prevents the development of the outcome at time 2, which can occur, for example, if asymptomatic infections are misclassified as no infections.

\begin{remark}[Null hypothesis under Assumption \ref{ASS: exp eff restriction}]\label{REM: Better null}
    Suppose that Assumptions \ref{ASS: Exp nec}, \ref{ASS: exp eff restriction}, and \ref{ASS: exp exch} hold. An alternative formulation of the effect restriction using nested counterfactuals is 
    $$
       \Delta Y_2^{a, e_1, e_2=1}=  \Delta Y_2^{a, e_1, \Delta Y_1^{a, e_1}, e_2=1}.
    $$
    By Assumption \ref{ASS: Exp nec} and \ref{ASS: exp exch}, $\Delta Y_1^{a, e_1=0}=0$, thus 
    \begin{align*}
        &\Delta Y_2^{a, e_1=0, e_2=1}=  \Delta Y_2^{a, e_1=0, \Delta Y_1^{a, e_1=0}, e_2=1}\\
        &=\Delta Y_2^{a, e_1=0, \delta y_1=0, e_2=1} = \Delta Y_2^{a, e_1=1, \delta y_1=0, e_2=1}. 
    \end{align*}
         
    Therefore, under  Assumptions \ref{ASS: Exp nec} and \ref{ASS: exp eff restriction}, Hypothesis \ref{HYP: formal null} is equivalent to 
    \begin{align*}
        &\Delta Y_1^{a, e_1=1} \overset{H_0}{=} 
        \Delta Y_2^{a, e_1=0, e_2=1} \\
        &=\Delta Y_2^{a, e_1=1, \delta y_1=0, e_2=1}=\Delta Y_2^{a, e_1=0, \delta y_1=0, e_2=1}.
    \end{align*}
\end{remark}

Suppose Hypothesis \ref{HYP: formal null} holds. Consider an individual who was exposed in the first interval but remained event-free. This individual will remain event-free at time 2 as well, regardless of the time 2 exposure. This fact is related to a constraint on the distribution of joint counterfactuals under Hypothesis \ref{HYP: formal null}, which we discuss in Appendix \ref{APP: PS general}.

\begin{remark}[Population- vs individual-level waning]\label{REM: waning and sharp}
    In principle, the challenge effects, given in Equations \eqref{EQ: VE_ch_1} and \eqref{EQ: VE_ch_2}, can be equal at times 1 and 2 even if $\Delta Y_1^{a, e_1=1}\neq \Delta Y_2^{a, e_1=0, e_2=1} w.p >0$. However, assume that
    $$\Delta Y_1^{a=0, e_1=1} = \Delta Y_2^{a=0, e_1=0, e_2=1} \; w.p.1, $$
    but not necessarily for $a=1$. Conceptually, this assumption states that the risk of infection under placebo does not depend on the time of exposure. An argument supporting this assumption is that no active ingredient is administered in the placebo arm, and thus, the immune system's response to exposure should not depend on the time to first exposure.  Then, the population-level notion of no vaccine waning is equivalent to the sharp null Hypothesis \ref{HYP: formal null}, if there is no perfect cancellation in the following sense: exactly the same number of people experience ``positive waning", that is, the improvement of the vaccine over time, as the number of people who experience ``negative waning", that is, the deterioration of the vaccine over time. Such a perfect cancellation of positive and negative effects is unlikely. Thus, we can often intuitively treat the sharp null of no waning, and population-level no waning ($VE_1^{\mathrm{challenge}}=VE_2^{\mathrm{challenge}})$ as equivalent. We provide a more formal argument building on the principal strata framework in Appendix \ref{APP: PS general}.
\end{remark}

The next results will lead to a simple testing strategy of the null hypothesis. To describe this result, consider first another assumption that is plausible in a blinded trial:

\begin{assumption}[No treatment effect on exposure in the unexposed]\label{ASS: blind exp}
    \begin{align*}
        & E_1^{a=0}=E_1^{a=1}&&\text{and} && E_2^{a=0, e_1=0}=E_2^{a=1, e_1=0}
    \end{align*}
\end{assumption}

Assumption \ref{ASS: blind exp} requires that there is no effect of treatment on time 1 exposure and on time 2 exposure following isolation. Conceptually, there is no causal path from $A$ to $E_1$. Furthermore, there is no causal path from $A$ to $E_2$ except through $E_1$ or $\Delta Y_1$. 
In a blinded trial, the exposure-seeking behavior at time 1 should not be a function of the actual treatment received.
Yet, at time 2, those who developed the outcome at time 1, might reduce their contacts, for example, to minimize the spread of the disease. If the vaccine has an effect on the outcome at time 1, then the time 2 exposure might decrease at an increased rate among the placebo recipients.
However, under isolation, the treatment would not affect their survival, by Assumption \ref{ASS: Exp nec}. Therefore, if the treatment is blinded, the time 2 exposure would be the same under both treatment and no treatment, making Assumption \ref{ASS: blind exp} plausible.

Assumption \ref{ASS: blind exp} is less restrictive than common assumptions in infectious disease epidemiology. In particular, random mixing is often assumed, which would require unconditional independence between the treatment and the exposures at both times, which implies Assumption \ref{ASS: blind exp}. Moreover, under random mixing, the two exposures cannot be confounded; this is a restriction that Assumption \ref{ASS: blind exp} does not impose. Consider also the following exchangeability condition:

\begin{assumption}[Exposure exchangeability]\label{ASS: exp exch}
    \begin{align*}
        E_1^{a} &\independent \Delta Y_2^{a, e_1=0}|E_2^{a, e_1=0}, &E_1^{a}&\independent \Delta Y_1^{a, e_1},\\
        E_2^{a, e_1=0} &\independent \Delta Y_2^{a, e_1=0, e_2}.
    \end{align*}
\end{assumption}

Exchangeabilities between the potential outcomes and the time 1 exposure are guaranteed by the lack of exposure-outcome confounders. However, in practice, exposure at time 1 and time 2 can be confounded by, for example, behavioral patterns. Thus, to make the independence between $E_1^{a}$ and $\Delta Y_2^{a, e_1=0}$ plausible, we have to control for the time 2 exposure.
Moreover, $E_2^{a, e_1=0}\independent \Delta Y_2^{a, e_1=0, e_2=1}$ is equivalent to $E_2^{a, e_1=0}\independent \Delta Y_2^{a, e_1=0, e_2=1}|\Delta Y_1^{a, e_1=0}$ under  Assumption \ref{ASS: Exp nec}. Conceptually, this independence holds in the absence of unmeasured exposure-outcome confounding. For ease of exposition, we present our results and assumptions without conditioning on measured covariates (confounders). However, the results generalize straightforwardly to settings with baseline covariates, which make the identifying assumptions more plausible in many applications, see Appendix \ref{APP:Cond wane} for details. Even when the measured baseline variables are conditioned on, some confounders that affect both the exposures and the outcomes can remain unmeasured. The more general Assumption \ref{ASS: cond exp exch} can be violated when we fail to adjust for such confounders. However, the resulting bias is expected to be modest when the influence or prevalence of the unmeasured confounder is limited.  When the investigators suspect substantial bias due to unmeasured confounding, they can mitigate such effects using negative control outcomes, such as non-vaccine-targeted infections, see \citet{ashby_debiasing_2025}. Adjustment for latent confounding is also motivated by the correlated-infections framework of \citet{farrington_correlated_2013} and \citet{unkel_time_2014}, which show that the correlations between infections within individuals can be informative about the extent of heterogeneity in effective contact rates, i.e, heterogeneity in exposure.

Under the directed acyclic graph (DAG) proposed by \citet{janvin_quantification_2024}, Assumption \ref{ASS: exp exch} holds. See Appendix \ref{APP: on Ass-s} for a replication of the DAG and a discussion on the Assumptions in \citet{janvin_quantification_2024}. 

The four assumptions introduced so far allow us to derive a simple equality that motivates our hypothesis test: 

\begin{proposition}[Invariant incidence rates]\label{PROP: inv incidence}
    Suppose that Assumptions \ref{ASS: Exp nec}-\ref{ASS: exp exch} hold. Under Hypothesis \ref{HYP: formal null}
    $$
        IR_1:= \frac{E[\Delta Y_1|A=1]}{E[\Delta Y_1|A=0]} = \frac{E[\Delta Y_2|A=1]}{E[\Delta Y_2|A=0]}:=IR_2.
    $$
\end{proposition}
See Proof in Appendix \ref{APP: Prop proof}.

Thus, Proposition \ref{PROP: inv incidence} allows us to test Hypothesis \ref{HYP: formal null} by comparing incidence ratios ($IR_1$, $IR_2$). Under Assumptions \ref{ASS: Exp nec}-\ref{ASS: exp exch}, the test is only informative about the existence of waning, but not the extent of it. However, if a further monotonicity assumption is imposed on the potential outcomes over time, a direct correspondence can be established between the test statistics and the ratio $\left(1-VE_2^{\mathrm{challenge}}\right)/\left(1-VE_1^{\mathrm{challenge}}\right)$, which provides a measure of the extent of waning. See Appendix \ref{APP: p-val} for details. 

We emphasize that incidence ratios are not hazard ratios, because time-varying hazard ratios (HRs) are often used in the literature when studying the waning of a vaccine. In a discrete-time setting, the observed hazard ratios are 
\begin{align*}
    &HR_1 := \frac{E[\Delta Y_1|A=1]}{E[\Delta Y_1|A=0]} && HR_2:=\frac{E[\Delta Y_2|\Delta Y_1=0, A=1]}{E[\Delta Y_2|\Delta Y_1=0, A=0]}. 
\end{align*}
At the first time, all individuals are event-free, regardless of which treatment arms they were randomized to. Thus, $IR_1=HR_1$. By Proposition \ref{PROP: inv incidence}, under the null hypothesis $HR_1=HR_2$, if and only if $P(\Delta Y_1=0|A=0)=P(\Delta Y_1=0|A=1)$, that is, the vaccine has no population-level effect on survival at the first time. Assuming that no effect of the vaccine is observed at the beginning of the study is unrealistic. Therefore, a change in the $HR$ over time, unlike a change in the $IR$ over time, is not a measure of waning. We quantify the bias of studying $HRs$ instead of $IRs$ in Appendix \ref{APP: Hazard}, and illustrate the performance of the hazard ratio-based testing in Appendix \ref{APP: Simulation}.

The test itself is suitable for assessing the presence of waning, however, it is generally uninformative of the magnitude of waning. In Appendix \ref{APP: alternative}, we discuss the interpretation of the ratio of the incidence ratios when the alternative of the null hypothesis holds, and possible improvements on the bounds of \citet{janvin_quantification_2024} when waning is present.

\section{Estimation and Data analysis}

The test can be implemented, using summary data only, see Appendix \ref{APP: Variance} for details. Thus, most published vaccine trials can be reanalyzed, given that they have the incidence recorded at least at two times. First, we revisit the BNT162b2 trial analyzed by \citet{janvin_quantification_2024} and show that our test is able to detect waning in the same setting in which they could not claim waning previously. Then we give an example where the test does not reject the null hypothesis, and argue that the test result agrees with biological explanations, illustrated by a \textit{Helicobacter pylori} vaccine study.

\citet{janvin_quantification_2024} have analyzed a blinded RCT, published by \citet{thomas_safety_2021}, in which trial participants were randomized to receive a placebo ($A=0$) or the BNT162b2 COVID-19 vaccine ($A=1$). We are following the same division of the follow-up to the first and second interval as \citet{janvin_quantification_2024} in their main analysis, and an identical estimation procedure of the conditional expectation $E[Y_k|A=a]$ for $a \in \{0,1\}, \, k \in\{1,2\}$, where $Y_k$ is the cumulative incidence by the end of interval $k$. Since their assumptions only allowed partial identification of $VE_2^{\mathrm{challenge}}$, they compared the point-identified $\widehat{VE}_1^{\mathrm{challenge}}$ to the estimated lower and upper bounds of the second interval challenge effects. $\widehat{\mathcal{L}}, \widehat{\mathcal{U}}$, respectively. Although their point estimates suggested a waning effect of the vaccine, the null hypothesis of no waning could not have been rejected.

In comparison, the point estimate of $\widehat{IR}_1/\widehat{IR}_2$ is $0.48$ $\big(P=0.01,\;  95\%\text{ CI } [0.279, 0.825]\big)$, thus clearly rejecting our null hypothesis of no change in the potential outcome of an individual upon exposure over time, at the $5\%$ level. Unless the condition in Remark \ref{REM: waning and sharp} holds, i.e., that the number of people for whom the vaccine’s effect improves over time equals the number for whom its efficacy decreases, the rejection can be interpreted as evidence of a change in the population challenge effect over time, i.e., waning vaccine efficacy.

Consider now a second example. \citet{zeng_efficacy_2015} performed a double-blind, placebo-controlled RCT, to assess the efficacy, safety and immunogenicity of the Helicobacter pylori vaccine among children in China. The trial participants were randomized 1:1 to the two arms and 99\% of them completed the three-dose vaccination schedule. The primary analysis was performed among these participants (per-protocol-analysis) to assess the H pylori infection in the first 12 months, following the completion of the vaccine schedule. 

Incidence data has been published at months 4,8,12,24, and 36. To compare with the primary analysis, we compared the incidence ratio until month 8 ($IR_1$), to the incidence ratio between month 8 and month 12 ($IR_2$). The choice of time intervals is application-specific and should be guided by the scientific question and, where relevant, public-health considerations. For example, decision-makers may wish to assess waning over periods aligned with feasible booster schedules. If booster doses can be administered every six months, one relevant comparison is vaccine efficacy during 0–6 months versus 6–12 months \citep{park_comparing_2024}. In the first period, there were 36 and 10 events for 1416 and 1403.6 person-years at risk in the placebo and treatment arms, respectively. Following that, in the second period, there were 14 and 4 events for 673.6 and 670.7 person-years at risk in the placebo and treatment arms, respectively. Using the direct $\delta$-method, we derive the estimate for $\widehat{IR}_1/\widehat{IR}_2 = 0.977$ ($P=0.97$ $95\%CI$ $[0.263, 3.63]$).

The test statistic agrees with the findings of \citet{zeng_efficacy_2015}; they have found that VE is $72\%\;(95\%CI$ $[42.4\%, 87.6\%])$ at month 8 and $71.8\%\;(95\%CI$ $[48.2\%, 85.6\%])$ at month 12. Even though we argued that the naive VE should not be used for assessing vaccine waning due to the selection-bias introduced, their analysis of the incidence also points towards no waning.

In practice, arguments about waning are often based on the decay of immunological responses, measured by, for example, antibody levels. \citet{zeng_efficacy_2015} measured the serum-specific antibodies in the vaccine and the placebo group. While at baseline the concentrations were similar in the two treatment groups, the vaccination elicited a strong immune response. The concentration of the serum-specific antibodies did not decrease significantly over time, agreeing with our test results.

\section{Discussion}
The challenge effect quantifies changes in the protective effect of the vaccine. However, it is not possible to point identify the challenge effect from conventional RCTs, and the sharp bounds of partial identification can often be too wide to make claims about waning.

We show that instead of identifying the challenge effect, a sharp null hypothesis of no waning can be assessed accurately, with minimal assumptions; that is, unless some unrealistic cancellations of waning hold, we can interpret the null hypothesis as equivalent to the population level notion of no waning, $VE_1^{\mathrm{challenge}}\neq VE_2^{\mathrm{challenge}}$. We derive a simple test statistic for the null hypothesis, which can be computed using the observed data from a standard RCT. In the absence of patient-level data, the test statistics can be computed using summary data only. 

The test is designed to assess the existence of waning. When waning is confirmed by the test, further steps can be taken to try to determine the magnitude of waning. This could, e.g., be a follow-up trial, a  reanalysis of the existing data with methods building on the results from \citet{janvin_quantification_2024}, or additional analyses leveraging expert knowledge on exposure probabilities, as discussed in Appendix \ref{APP_sub: Challenge sens}.

When evidence of waning remains inconclusive, this may reflect the limited duration of follow-up. Such limitations are common in vaccine trials, where the blinded phase is often discontinued once early efficacy has been established. A next step is to extend the proposed testing strategy to incorporate information from, for example, post-unblinding and open-label follow-up period \citep{fintzi_assessing_2021}, or to adapt the strategy to modified trial designs, for example, inspired by \citet{follmann_deferred-vaccination_2021}. These extensions can allow us to relax some assumptions and improve power by leveraging additional cases.

\bibliographystyle{unsrtnat}
\bibliography{references}

\section*{Online Appendix}
\addcontentsline{toc}{section}{Appendices}
\renewcommand{\thesubsection}{\Alph{subsection}}

\subsection{Principal Strata at Different Times}\label{APP: PS general}
Suppose that we are conducting a 4-arm trial, as described in Section \ref{SEC: Intro} and depicted in Figure \ref{FIG: 4-arms} of the main manuscript. Consider a counterfactual causal model such that there exists a pre-treatment variable, denoted by $T \in \{1, \dots, 16\}$, such that
    \begin{align*}
        T = 1 &\implies (\Delta Y_1^{a=1, e_1=1}=0, \Delta Y_1^{a=0, e_1=1}=0, \Delta Y_2^{a=1, e_1=0, e_2=1}=0, \Delta Y_2^{a=0, e_1=0, e_2=1}=0)\\
        T = 2 &\implies (\Delta Y_1^{a=1, e_1=1}=0, \Delta Y_1^{a=0, e_1=1}=0, \Delta Y_2^{a=1, e_1=0, e_2=1}=0, \Delta Y_2^{a=0, e_1=0, e_2=1}=1)\\
        T = 3 &\implies (\Delta Y_1^{a=1, e_1=1}=0, \Delta Y_1^{a=0, e_1=1}=0, \Delta Y_2^{a=1, e_1=0, e_2=1}=1, \Delta Y_2^{a=0, e_1=0, e_2=1}=0)\\
        T = 4 &\implies (\Delta Y_1^{a=1, e_1=1}=0, \Delta Y_1^{a=0, e_1=1}=0, \Delta Y_2^{a=1, e_1=0, e_2=1}=1, \Delta Y_2^{a=0, e_1=0, e_2=1}=1)\\
        T = 5 &\implies (\Delta Y_1^{a=1, e_1=1}=0, \Delta Y_1^{a=0, e_1=1}=1, \Delta Y_2^{a=1, e_1=0, e_2=1}=0, \Delta Y_2^{a=0, e_1=0, e_2=1}=0)\\
        T = 6 &\implies (\Delta Y_1^{a=1, e_1=1}=0, \Delta Y_1^{a=0, e_1=1}=1, \Delta Y_2^{a=1, e_1=0, e_2=1}=0, \Delta Y_2^{a=0, e_1=0, e_2=1}=1)\\
        T = 7 &\implies (\Delta Y_1^{a=1, e_1=1}=0, \Delta Y_1^{a=0, e_1=1}=1, \Delta Y_2^{a=1, e_1=0, e_2=1}=1, \Delta Y_2^{a=0, e_1=0, e_2=1}=0)\\
        T = 8 &\implies (\Delta Y_1^{a=1, e_1=1}=0, \Delta Y_1^{a=0, e_1=1}=1, \Delta Y_2^{a=1, e_1=0, e_2=1}=1, \Delta Y_2^{a=0, e_1=0, e_2=1}=1)\\
        T = 9 &\implies (\Delta Y_1^{a=1, e_1=1}=1, \Delta Y_1^{a=0, e_1=1}=0, \Delta Y_2^{a=1, e_1=0, e_2=1}=0, \Delta Y_2^{a=0, e_1=0, e_2=1}=0)\\
        T = 10 &\implies (\Delta Y_1^{a=1, e_1=1}=1, \Delta Y_1^{a=0, e_1=1}=0, \Delta Y_2^{a=1, e_1=0, e_2=1}=0, \Delta Y_2^{a=0, e_1=0, e_2=1}=1)\\
        T = 11&\implies (\Delta Y_1^{a=1, e_1=1}=1, \Delta Y_1^{a=0, e_1=1}=0, \Delta Y_2^{a=1, e_1=0, e_2=1}=1, \Delta Y_2^{a=0, e_1=0, e_2=1}=0)\\
        T = 12 &\implies (\Delta Y_1^{a=1, e_1=1}=1, \Delta Y_1^{a=0, e_1=1}=0, \Delta Y_2^{a=1, e_1=0, e_2=1}=1, \Delta Y_2^{a=0, e_1=0, e_2=1}=1)\\
        T = 13 &\implies (\Delta Y_1^{a=1, e_1=1}=1, \Delta Y_1^{a=0, e_1=1}=1, \Delta Y_2^{a=1, e_1=0, e_2=1}=0, \Delta Y_2^{a=0, e_1=0, e_2=1}=0)\\
        T = 14 &\implies (\Delta Y_1^{a=1, e_1=1}=1, \Delta Y_1^{a=0, e_1=1}=1, \Delta Y_2^{a=1, e_1=0, e_2=1}=0, \Delta Y_2^{a=0, e_1=0, e_2=1}=1)\\
        T = 15 &\implies (\Delta Y_1^{a=1, e_1=1}=1, \Delta Y_1^{a=0, e_1=1}=1, \Delta Y_2^{a=1, e_1=0, e_2=1}=1, \Delta Y_2^{a=0, e_1=0, e_2=1}=0)\\
        T = 16 &\implies (\Delta Y_1^{a=1, e_1=1}=1, \Delta Y_1^{a=0, e_1=1}=1, \Delta Y_2^{a=1, e_1=0, e_2=1}=1, \Delta Y_2^{a=0, e_1=0, e_2=1}=1)
    \end{align*}

The counterfactual groups above are a composite of the 4 conventional deterministic principal strata defined at each point in time. For example, someone may be doomed at time 1, while immune at time 2, which corresponds to $T=4$. The transition from one principal stratum to another corresponds to a waning at the individual level, as one's response to exposure changes, at least under one of the two treatment levels.

In this general setting, the challenge effects are 
\begin{align*}
    VE_1^{\mathrm{challenge}} &=  \frac{P\big(T \in \{9, \dots, 16\}\big)}{P\big(T \in \{5, \dots, 8,13, \dots, 16\}\big)}\\
    VE_2^{\mathrm{challenge}}&= \frac{P\big(T \in \{3,4, 7,8, 11, 12, 15,16\}\big)}{P\big(T \in \{2,4,6,8,10,12,14,16\}\big)},
\end{align*}

which are equal to the observed incidence ratios under a suitable set of assumptions, as shown in Appendix \ref{APP: Prop proof}.

\begin{corollary}[Principal Strata under the null]
    Suppose Hypothesis \ref{HYP: formal null} holds. Then
    $$P(T \in \{1, 6, 11, 16\}) = 1$$
\end{corollary}
\begin{proof}
Under the null hypothesis, and by partitioning
    \begin{align*}
        0&\overset{H_0}{=}P(\Delta Y_1^{a=1, e_1=1}=0, \Delta Y_2^{a=1, e_1=0, e_2=1}=1)\\
        &= P(\Delta Y_1^{a=1, e_1=1}=0, \Delta Y_1^{a=0, e_1=1}=0, \Delta Y_2^{a=1, e_1=0, e_2=1}=1, \Delta Y_2^{a=0, e_1=0, e_2=1}=0)\\
        &+P(\Delta Y_1^{a=1, e_1=1}=0, \Delta Y_1^{a=0, e_1=1}=0, \Delta Y_2^{a=1, e_1=0, e_2=1}=1, \Delta Y_2^{a=0, e_1=0, e_2=1}=1)\\
        &+P(\Delta Y_1^{a=1, e_1=1}=0, \Delta Y_1^{a=0, e_1=1}=1, \Delta Y_2^{a=1, e_1=0, e_2=1}=1, \Delta Y_2^{a=0, e_1=0, e_2=1}=0)\\
        &+P(\Delta Y_1^{a=1, e_1=1}=0, \Delta Y_1^{a=0, e_1=1}=1, \Delta Y_2^{a=1, e_1=0, e_2=1}=1, \Delta Y_2^{a=0, e_1=0, e_2=1}=1).
    \end{align*}
    Therefore, by the axioms of probability $P(T=3)=P(T=4)=P(T=7)=P(T=8)=0$. Similarly, under the null 
    \begin{align*}
        &P(\Delta Y_1^{a=1, e_1=1}=1, \Delta Y_2^{a=0, e_1=0, e_2=1}=0)=P(\Delta Y_1^{a=0, e_1=1}=0, \Delta Y_2^{a=0, e_1=0, e_2=1}=1)\\
        &=P(\Delta Y_1^{a=0, e_1=1}=1, \Delta Y_2^{a=0, e_1=0, e_2=1}=0)=0,
    \end{align*}
    which implies that
    \begin{align*}
        &P(T=9)=P(T=10)=P(T=13)=P(T=14)=0\\
        &P(T=2)=P(T=4)=P(T=10)=P(T=12)=0\\
        &P(T=5)=P(T=7)=P(T=13)=P(T=15)=0.
    \end{align*}
    Thus, we have that $P(T \in \{2,3,4,5,7,8,9,10,12,13,14,15\})=0$, which means that $P(T\in \{1,6,11,16\})=1$.
\end{proof}

It follows that, under the null hypothesis, the challenge effects are equal to 
\begin{align*}
    VE_1^{\mathrm{challenge}} &=  \frac{P\big(T \in \{9, \dots, 16\}\big)}{P\big(T \in \{5, \dots, 8,13, \dots, 16\}\big)}\overset{H_0}{=}\frac{P\big(T \in \{11, 16\}\big)}{P\big(T \in \{6,16\}\big)} \\
    VE_2^{\mathrm{challenge}}&= \frac{P\big(T \in \{3,4, 7,8, 11, 12, 15,16\}\big)}{P\big(T \in \{2,4,6,8,10,12,14,16\}\big)} \overset{H_0}{=}\frac{P\big(T \in \{11, 16\}\big)}{P\big(T \in \{6,16\}\big)},
\end{align*}
which shows their equality under the null.

We argue for the equivalence of the population level and the sharp no waning, under the sharp no waning of the placebo, that is, $\Delta Y_1^{a=0, e_1=1}= \Delta Y_2^{a=0, e_1=0, e_2=1}$ with probability $1$. We believe this is a plausible assumption, as the effect of the placebo should not change over time for an individual if they are isolated; hence, their immune system is not challenged before the exposure. Therefore, if they were to develop the outcome upon exposure at time 1, then they would develop it at time 2, given that the exposures are identical. Similarly, if they were to remain event-free at time 1 due to natural immunity, for example, due to genetic reasons, a quality that would not change over time and they would remain event-free at time 2 as well.

By identical arguments as in the proof above, the sharp no waning of the placebo implies that
$$
P(T \in \{2,4,5,7,10,12,13,15\}) =0.
$$
In this setting, the challenge effects are
\begin{align*}
    VE_1^{\mathrm{challenge}} &=  \frac{P\big(T \in \{9, \dots, 16\}\big)}{P\big(T \in \{5, \dots, 8,13, \dots, 16\}\big)}=\frac{P\big(T \in \{9, 11, 14, 16\}\big)}{P\big(T \in \{6, 8, 14, 16\}\big)} \\
    VE_2^{\mathrm{challenge}}&= \frac{P\big(T \in \{3,4, 7,8, 11, 12, 15,16\}\big)}{P\big(T \in \{2,4,6,8,10,12,14,16\}\big)} =\frac{P\big(T \in \{3,8,11,16\}\big)}{P\big(T \in \{6,8, 14,16\}\big)}.
\end{align*}
Therefore, $VE_1^{\mathrm{challenge}}=VE_2^{\mathrm{challenge}}$, when the null hypothesis does not hold, if and only if $\big(T \in \{9, 11, 14, 16\}\big)=\big(T \in \{3,8,11, 16\}\big)$, that is equivalent to 
$$
    P(T=9)+P(T=14)=P(T=3)+P(T=8).
$$

Here, $T=9$ corresponds to individuals who are harmed by vaccination in the first interval but who would have been immune at time 2 had they instead been isolated in prior to that. In contrast, $T=14$ denotes individuals who are doomed in the first interval but who would have been helped in the second interval under following isolation. Thus, $T=9 \cup T=14$ comprises those whose potential outcome under vaccination improves from time 1 to time 2.

Similarly, $T=3 \cup T=8$ comprises individuals who transition from immune to harmed or from helped to doomed, respectively; that is, those whose potential outcome under vaccination worsens from time 1 to time 2. Informally, the equality above implies that, in the vaccinated arm, the number of individuals experiencing \textit{positive} waning is equal to the number experiencing \textit{negative} waning.

In general, it is unlikely that the vaccine effect improves over time, i.e., $P(T \in \{9, 14\}) \approx 0$. Even if such improvement were possible, an exact balance between vaccine improvement and deterioration would be even less plausible. Therefore, the null hypothesis of no individual-level waning is plausibly equivalent to the null hypothesis of no population-level waning.

\subsection{Proof of Proposition \ref{PROP: inv incidence}}\label{APP: Prop proof}
\begin{proof}
    The observed incidence rate ratio at time 1 is equal to 
\begin{align*}
    1-VE_1^{\mathrm{challenge}} &:= \frac{E[\Delta Y_1^{a=1, e_1=1}]}{E[\Delta Y_1^{a=0, e_1=1}]}\\
    &\overset{\mathrm{Assumption } \; \ref{ASS: exp exch}}{=}\frac{E[\Delta Y_1^{a=1, e_1=1}|E_1^{a=1}=1]}{E[\Delta Y_1^{a=0, e_1=1}|E_1^{a=0}=1]}\\
    &\overset{\mathrm{Assumption } \; \ref{ASS: Cons}}{=}\frac{E[\Delta Y_1^{a=1}|E_1^{a=1}=1]}{E[\Delta Y_1^{a=0}|E_1^{a=0}=1]}\\
    &\overset{\mathrm{Assumption } \; \ref{ASS: Exp nec}}{=}\frac{E[\Delta Y_1^{a=1}]/P(E_1^{a=1}=1)}{E[\Delta Y_1^{a=0}]/P(E_1^{a=0}=1)}\\
    &\overset{\mathrm{Assumption } \; \ref{ASS: blind exp}}{=}\frac{E[\Delta Y_1^{a=1}]}{E[\Delta Y_1^{a=0}]}\\
    &\overset{\mathrm{Assumption } \; \ref{ASS: Trt exch}, \, \ref{ASS: pos}}{=}\frac{E[\Delta Y_1^{a=1}|A=1]}{E[\Delta Y_1^{a=0}| A=0]}\\
    &\overset{\mathrm{Assumption }\; \ref{ASS: Cons}}{=}\frac{E[\Delta Y_1|A=1]}{E[\Delta Y_1| A=0]}
\end{align*}

At time 2, the incidence rate for the arm $A=a$ is
\begin{align*}
    E[\Delta Y_2|A=a]&\overset{\mathrm{Assumption} \; \ref{ASS: Cons}, \, \ref{ASS: Trt exch}, \, \ref{ASS: pos}} {=} E[\Delta Y_2^{a}] \\
    &=E[\Delta Y_2^{a}|E_1^{a}=0]P(E_1^{a}=0)\\
    &+E[\Delta Y_2^{a}| E_1^{a}=1]P(E_1^{a}=1).\\
\end{align*}
Considering the first term
\begin{align*}
    E[\Delta Y_2^{a}|E_1^{a}=0]P(E_1^{a}=0) &\overset{\mathrm{Assumption}\; \ref{ASS: Exp nec}}{=}  E[\Delta Y_2^{a}|E_1^{a}=0, E_2^{a}=1]P(E_2^{a}=1|E_1^{a}=0)P(E_1^{a}=0)\\
    & \overset{\mathrm{Assumption}\;\ref{ASS: Cons}}{=} E[\Delta Y_2^{a, e_1=0}|E_1^{a}=0, E_2^{a, e_1=0}=1]P(E_2^{a, e_1=0}=1|E_1^{a}=0)P(E_1^{a}=0)\\
    &\overset{\mathrm{Assumption}\; \ref{ASS: exp exch}}{=}  E[\Delta Y_2^{a, e_1=0}| E_2^{a, e_1=0}=1]P(E_2^{a, e_1=0}=1|E_1^{a}=0)P(E_1^{a}=0)\\
    &\overset{\mathrm{Assumption}\; \ref{ASS: Cons}}{=}  E[\Delta Y_2^{a, e_1=0, e_2=1}|E_2^{a, e_1=0}=1]P(E_2^{a, e_1=0}=1, E_1^{a}=0)\\
    &\overset{\mathrm{Assumption}\; \ref{ASS: exp exch}}{=}  E[\Delta Y_2^{a, e_1=0, e_2=1}]P(E_2^{a, e_1=0}=1, E_1^{a}=0).
\end{align*}

The second term equals
\begin{equation}\label{EQ: time1 exposed}
    \begin{aligned}
    E[\Delta Y_2^{a}|E_1^{a}=1]P(E_1^{a}=1) &=E[\Delta Y_2^{a}|E_1^{a}=1, \Delta Y_1^{a}=0]P(\Delta Y_1^{a}=0|E_1^{a}=1)P(E_1^{a}=1)\\
    &\overset{\mathrm{Assumption} \; \ref{ASS: Exp nec}}{=}E[\Delta Y_2^{a}|E_1^{a}=1, \Delta Y_1^{a}=0, E_2^{a}=1]\\
    &\times P(E_2^{a}=1|E_1^{a}=1, \Delta Y_1^{a}=0) P(\Delta Y_1^{a}=0|E_1^{a}=1)P(E_1^{a}=1)\\
    & \overset{\mathrm{Assumption}\;\ref{ASS: Cons}}{=} E[\Delta Y_2^{a, e_1=1}|E_1^{a}=1, \Delta Y_1^{a, e_1=1}=0, E_2^{a, e_1=1}=1]\\
    &\times P(E_2^{a, e_1=1}=1|E_1^{a}=1, \Delta Y_1^{a, e_1=1}=0) P(\Delta Y_1^{a, e_1=1}=0|E_1^{a}=1)P(E_1^{a}=1)
\end{aligned}
\end{equation}

Using Assumption \ref{ASS: Cons}
\begin{align}
    E&[\Delta Y_2^{a, e_1=1}|E_1^{a}=1, \Delta Y_1^{a, e_1=1}=0, E_2^{a, e_1=1}=1]\label{EQ: cross-trick}\\
    &=E[\Delta Y_2^{a, e_1=1, \delta y_1=0, e_2=1}|E_1^{a}=1, \Delta Y_1^{a, e_1=1}=0, E_2^{a, e_1=1, \delta y_1=0}=1].
\end{align}

As shown in Remark \ref{REM: Better null}, under Assumption \ref{ASS: exp eff restriction} and the null hypothesis, $\Delta Y_2^{a, e_1=1, \delta y_1=0, e_2=1}=\Delta Y_1^{a, e_1=1}$. Therefore
$$
    E[\Delta Y_2^{a, e_1=1, \delta y_1=0, e_2=1}|E_1^{a}=1, \Delta Y_1^{a, e_1=1}=0, E_2^{a, e_1=1, \delta y_1=0}=1] =0,
$$
which shows that $ E[\Delta Y_2^{a}|E_1^{a}=1]P(E_1^{a}=1) =0$. Finally, 
\begin{align*}
    \frac{E[\Delta Y_2|A=1]}{E[\Delta Y_2|A=1]} &=\frac{E[\Delta Y_2^{a=1}|E_1^{a=1}=0]P(E_1^{a=1}=0)}{E[\Delta Y_2^{a=0}|E_1^{a=0}=0]P(E_1^{a=0}=0)}\\
    &=\frac{E[\Delta Y_2^{a=1, e_1=0, e_2=1}]P(E_2^{a=1, e_1=0}=1, E_1^{a=1}=0)}{E[\Delta Y_2^{a=0, e_1=0, e_2=1}]P(E_2^{a=0, e_1=0}=1, E_1^{a=0}=0)}\\
    &\overset{\mathrm{Assumption} \; \ref{ASS: blind exp}}{=}\frac{E[\Delta Y_2^{a=1, e_1=0, e_2=1}]}{E[\Delta Y_2^{a=0, e_1=0, e_2=1}]}:=1- VE_2^{\mathrm{challenge}}
\end{align*}

Under the null hypothesis $E[\Delta Y_1^{a, e_1=1}]=E[\Delta Y_2^{a, e_1=0, e_2=1}]$, which implies $VE_1^{\mathrm{challenge}}=VE_2^{\mathrm{challenge}}$, and thus the equality $IR_1=IR_2$ holds.
\end{proof}

\subsection{Simulation}\label{APP: Simulation}
Consider the following data-generating mechanism
\begin{align*}
    A &:= \mathrm{Bernoulli}(0.5)\\
    T_1 &:=\mathrm{Categorical}(p_{doomed_1}, p_{helped_1}, p_{harmed_1}, p_{immune_1})\\
    E_1 &:= \mathrm{Bernoulli}(p_{E_1})\\
    \Delta Y_1 &:= E_1 \times \big(I(T_1=1)+(1-A)\times I(T_1=2)+A\times I(T_1=3)\big)\\
    E_2 &:= (1-\Delta Y_1) \times \mathrm{Bernoulli}(p_{E_2})\\
    \Delta Y_2 &:= E_2 \times \big(I(T_2=1)+(1-A)\times I(T_2=2)+A\times I(T_2=3)\big).
\end{align*}
Here, $T_1$ corresponds to the 4 standard principal strata, as discussed in Appendix \ref{APP: PS general}. While the proof and the assumptions of Proposition \ref{PROP: inv incidence} do not rely on the notion of the principal strata, to simulate the observed outcomes while respecting the constraints imposed on the challenge effects commonly involves assigning to all individuals their potential outcomes under all treatment levels and realizing their observed data by simulating the observed treatment.
Furthermore, $T_2$ is determined in a way to decrease the first interval challenge effect $VE_1^{\mathrm{challenge}}=1-\frac{p_{doomed_1}+p_{harmed_1}}{p_{doomed_1}+p_{helped_1}}$ such that we randomly resample a proportion of the helped individuals as doomed, to satisfy the assumption of sharp no waning of the placebo. The proportion is calculated, such that it is ensured in expectation that $\frac{p_{doomed_1}+p_{harmed_1}}{p_{doomed_1}+p_{helped_1}}\times w = \frac{p_{doomed_2}+p_{harmed_2}}{p_{doomed_2}+p_{helped_2}}$ for some factor $w>1$, (with $w=1$ corresponding to the sharp null hypothesis) (see Appendix \ref{APP_sub: wane spec} for details).
$T=1,2,3,4$ are defined to correspond to being \textit{doomed}, \textit{helped}, \textit{harmed}, and \textit{immune}, respectively. 

\begin{figure}[!ht]
    \centering
    \includegraphics[width=0.9\linewidth]{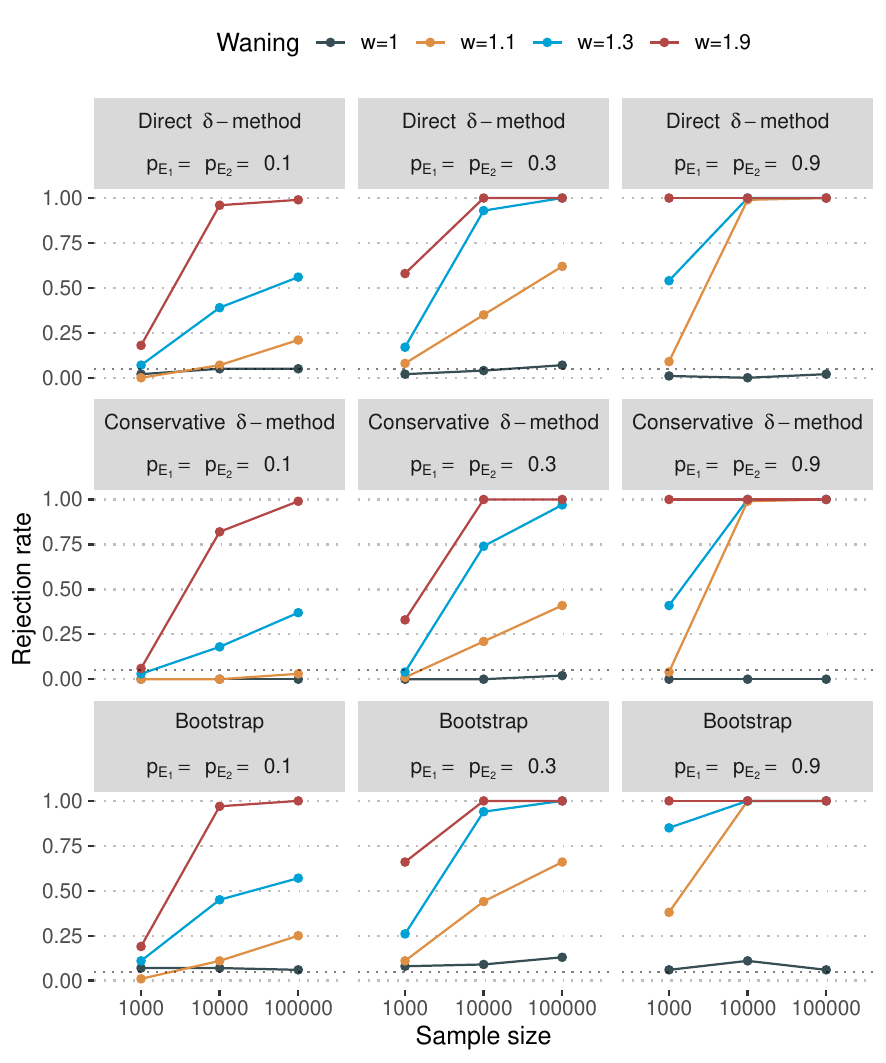}
    \caption{Rejection rate of the test statistics at level $\alpha =0.05$ from 100 simulations. The rows correspond to the direct $\delta$-method, the conservative $\delta$-method, and the non-parametric bootstrap, respectively.}
    \label{FIG: waning rejection rate}
\end{figure}

By Proposition \ref{PROP: inv incidence}, the two incidence ratios are equal when the sharp null hypothesis holds. Therefore, we can test whether $IR_1/IR_2 = 1$ to assess the null hypothesis. 

We propose three alternative methods to construct the confidence interval for the fraction of the incidence ratios. We can approximate the four empirical means, corresponding to the expected outcome at the first and the second time under treatment vs no treatment, respectively, as log-normal distributions. We can calculate the variance of both $IR_1$ and $IR_2$ using the $\delta$-method; however, the incidence ratios are negatively correlated, which depends on the level of exposure and the distribution of the susceptibility. To have a valid confidence interval, the correlation can be lower bounded by $-1$, but this is a considerable over-approximation when exposure is low or moderate, and when the population is heterogeneous in terms of susceptibility. To provide more precise confidence intervals, we can use the $\delta$-method to derive a direct estimator for the variance of $\mathrm{log}(\widehat{IR_1})-\mathrm{log}(\widehat{IR_2})$. Details of both are provided in Appendix \ref{APP: Variance}. Finally, given that we have access to patient-level data, we can derive the bounds using non-parametric bootstrap methods \citep{davison_bootstrap_1997}.

Figure \ref{FIG: waning rejection rate} illustrates the performance of the derived test statistics using the three alternatives. We set the baseline vaccine effect to $VE_1^{\mathrm{challenge}}=0.5$, by simulating $T_1$ from $\mathrm{Categorical}(0.2, 0.6, 0.2, 0)$. Under the varying level of exposure $E_1=E_2 \in \{0.1, 0.3, 0.9\}$, the level of overestimation of the variance is changing as illustrated in the first row. In particular, for relatively low exposure, $E_1=E_2=0.1$, the test is more conservative, and for low sample size, the null hypothesis is rejected at the level of the test, even when waning is substantial, that is, the vaccine effect decreases to $VE=0.05$. However, when the correlation approaches the supposed lower bound, by $E_1=E_2=0.9$, the null hypothesis is rejected at a rate of $0.54$, even for a moderate sample size of $n=1000$, and limited waning, $VE_1=0.5$, $VE_2=0.35$. 
In any of the simulated scenarios, when the null hypothesis holds, the test rejects at a rate below its level. 

Deriving the confidence interval using non-parametric bootstrap has greater power than using the direct $\delta$-method. However, the $\delta$-method maintains the nominal level, while the bootstrap shows deviations from it.

Extension to the time-to-event framework is trivial, except that the waning cannot be entirely continuous; the transition from one's reaction to exposure to a different reaction can be achieved in a discrete manner. Similar to the original paper of \citet{janvin_quantification_2024}, we can estimate the event rates using the cumulative hazard. Whether exposure is recurring or a single one does not matter, since under the null hypothesis, the outcome will be deterministically the same following each exposure. 

As discussed in Section \ref{SEC: Inference}, the equality of the hazard ratios is unsuitable for testing the waning of the vaccine effect, informally, due to the depletion of susceptibles over time. Using the same simulated data, we constructed confidence intervals for $\widehat{HR}_1/\widehat{HR}_2$ using non-parametric bootstrap, to test $HR_1/HR_2=1$.

\begin{figure}
    \centering
    \includegraphics[width=0.8\linewidth]{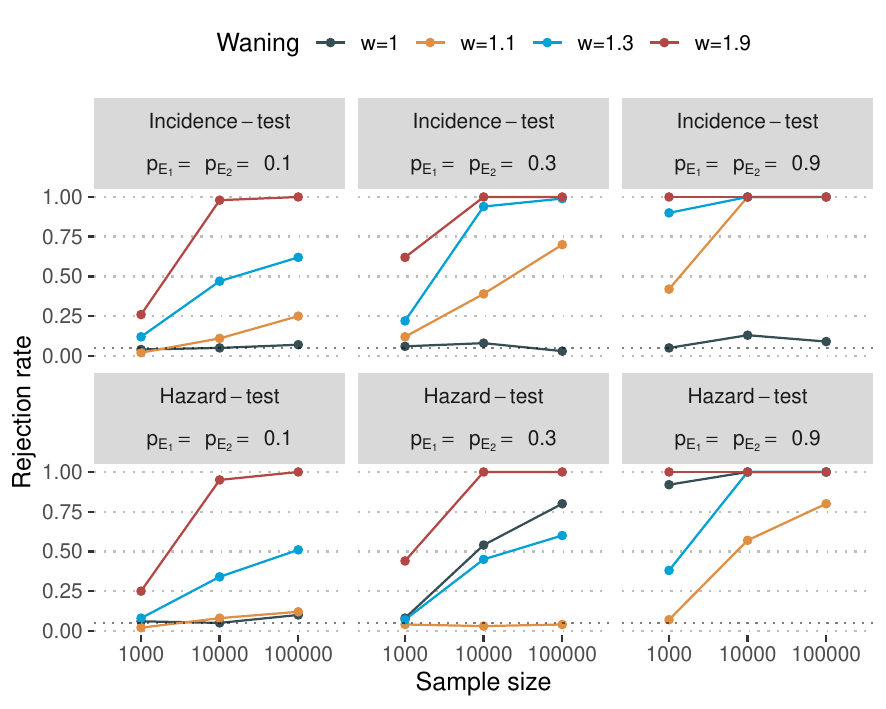}
    \caption{Rejection rates of the test statistics for $\widehat{HR}_1/\widehat{HR}_2$ at level $\alpha=0.05$, compared to $IR$ based test statistics, from $100$ simulations, using non-parametric bootstrap.}
    \label{FIG: HR test}
\end{figure}

Figure \ref{FIG: HR test} shows that when exposure is low, and hence the depletion of susceptibles is minor, the test based on hazard ratios is comparable to the test using incidence ratios, see Subfigure \ref{FIG: waning rejection rate} C. However, when the level of infectious pressure increases by exposure, when the null hypothesis holds ($w=0$), the test does not hold its nominal level. In addition, the power of the test does not increase with the level of the waning, which arguably is a desirable property of any sensible test of waning.

\subsection{Sensitivity analysis}\label{APP: Sens}
Consider the same data-generating mechanism as in the main part of Appendix \ref{APP: Simulation}, except for $T_1$ and $T_2$. As opposed to the simulation shown in the main part of the manuscript, suppose
$$
T_1 \sim \mathrm{Categorical}(0.2, 0.4, 0.1, 0.3).
$$

Under the assumption of sharp no waning of the placebo, we will consider three different scenarios for determining $T_2$, when the individual-level vaccine effect decreases:
\begin{itemize}
    \item Transitioning a proportion of helped-to-doomed ($w^{*}_{immune}=0$), which corresponds to the setting depicted in Figure \ref{FIG: waning rejection rate}, however, under a different distribution of $T_1$.
    \item Transitioning a proportion of immune-to-harmed ($w^{*}_{helped}=0$).
    \item Transitioning an equal proportion of helped and immune to doomed and harmed, respectively ($w^{*}_{helped}=w^{*}_{immune}$).
\end{itemize}
For details of determining the exact proportions $w^*$ to match the desired amount of waning $w$, see Appendix \ref{APP_sub: wane spec}.

The simulation results are shown in Figures \ref{FIG: App sens help}, \ref{FIG: App sens immune}, and \ref{FIG: App sens both}.

\begin{figure}[ht!]
    \centering
    \includegraphics[width=0.75\linewidth]{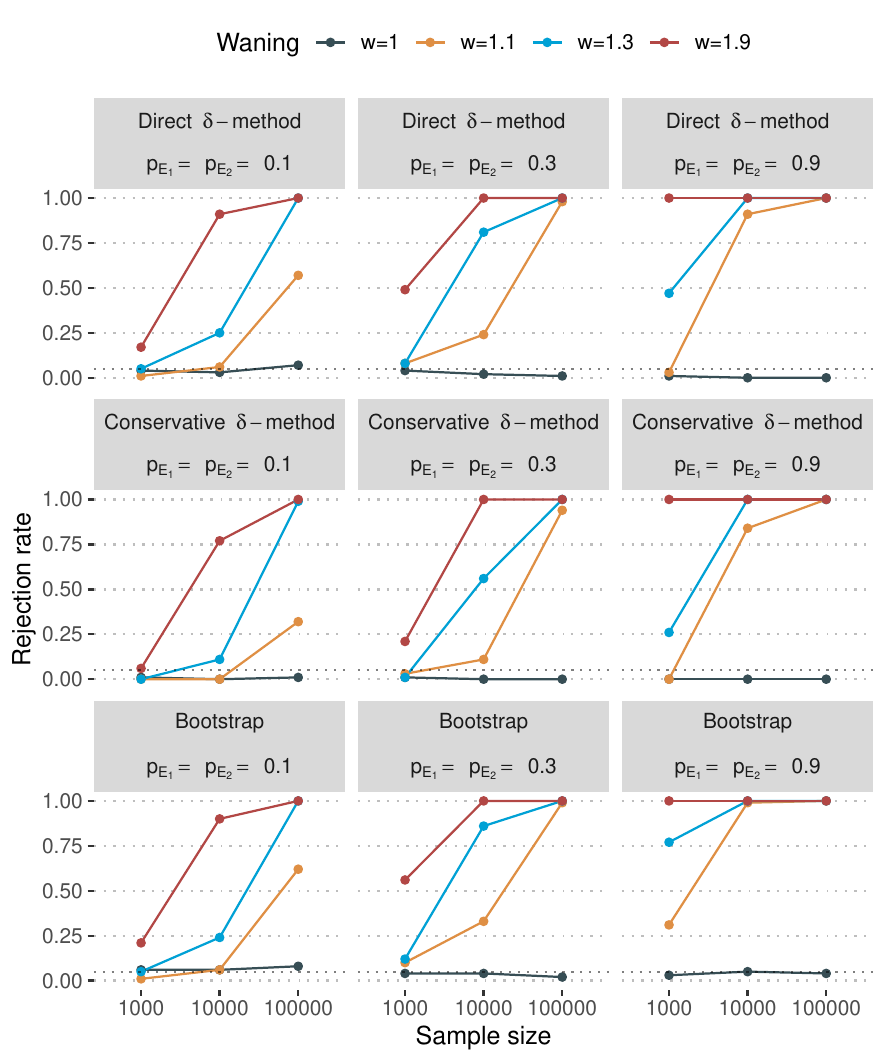}
    \caption{Rejection rates under helped-to-doomed transitions.}
    \label{FIG: App sens help}
\end{figure}

\begin{figure}[ht!]
    \centering
    \includegraphics[width=0.75\linewidth]{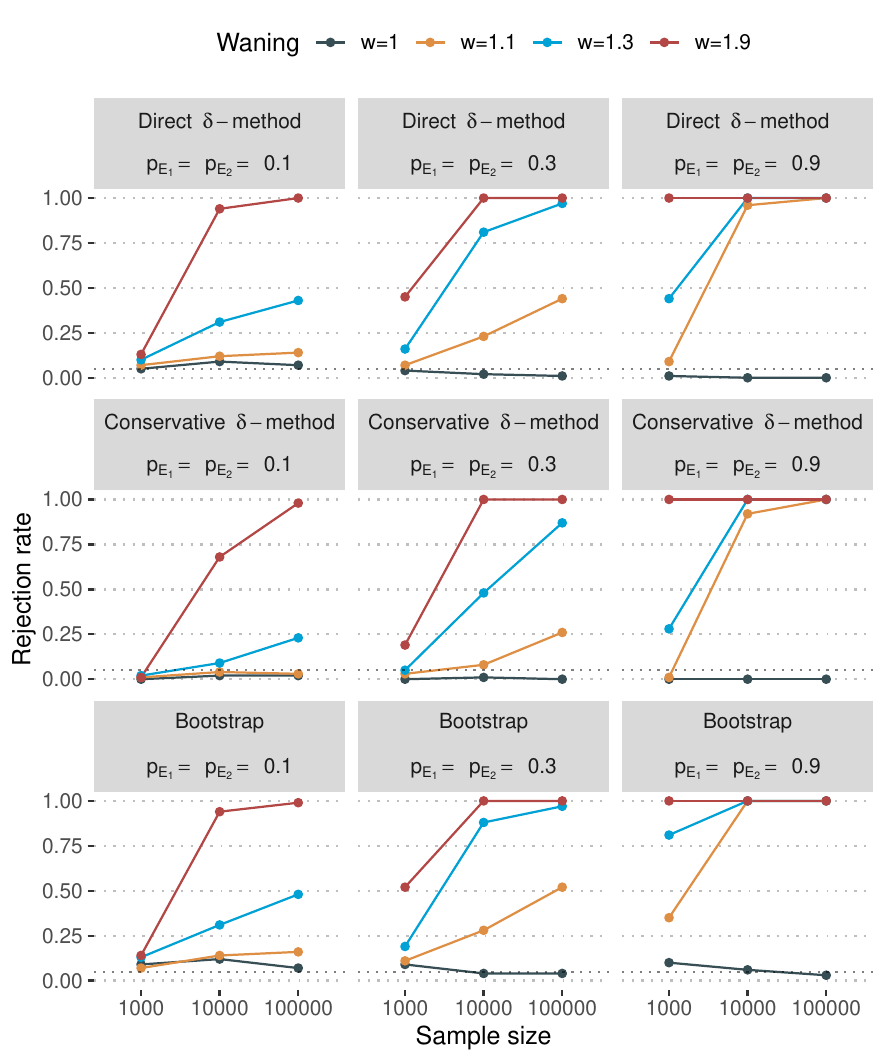}
    \caption{Rejection rates under immune-to-harmed transitions.}
    \label{FIG: App sens immune}
\end{figure}

\begin{figure}[ht!]
    \centering
    \includegraphics[width=0.75\linewidth]{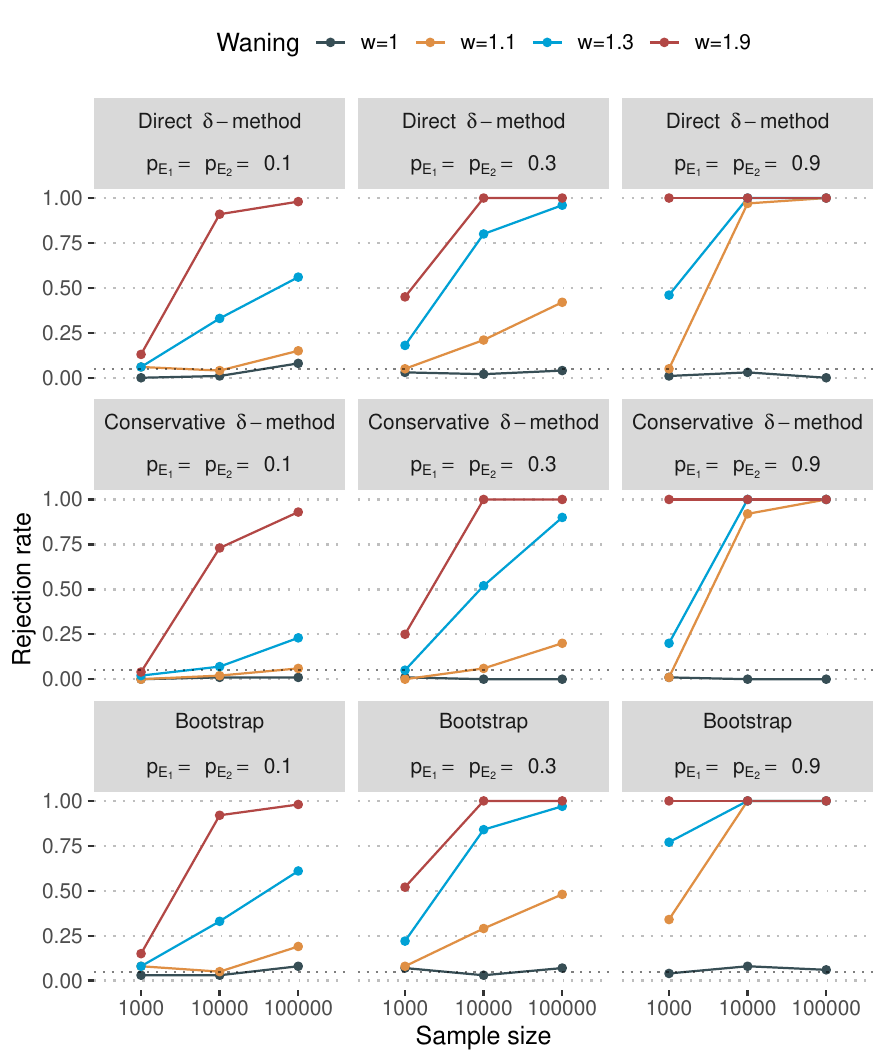}
    \caption{Rejection rates under equal helped-to-doomed and immune-to-harmed transitions.}
    \label{FIG: App sens both}
\end{figure}

When we include immune individuals in our cohort, by lowering the incidence rates of the outcomes, the variance of the estimator increases.

Comparing Figure \ref{FIG: App sens help} to \ref{FIG: App sens immune}, we can see that restricting the waning to helped-to-doomed,

Finally, even though we allow for unmeasured confounding between the exposures over time, for simplicity, we omitted that from the simulations. Consider the modified data-generating mechanism 
\begin{align*}
    A &:= \mathrm{Bernoulli}(0.5)\\
    U_E &:= \mathrm{Bernoulli}(0.5)\\
    T_1 &:=\mathrm{Categorical}(p_{doomed_1}, p_{helped_1}, p_{harmed_1}, p_{immune_1})\\
    E_1 &:= U_E\times \mathrm{Bernoulli}(p_{E_{high}})+(1-U_E)\times \mathrm{Bernoulli}(p_{E_{low}})\\
    \Delta Y_1 &:= E_1 \times \big(I(T_1=1)+(1-A)\times I(T_1=2)+A\times I(T_1=3)\big)\\
    E_2 &:= (1-\Delta Y_1) \times \big(U_E\times \mathrm{Bernoulli}(p_{E_{high}})+(1-U_E)\times \mathrm{Bernoulli}(p_{E_{low}})\big)\\
    \Delta Y_2 &:= E_2 \times \big(I(T_2=1)+(1-A)\times I(T_2=2)+A\times I(T_2=3)\big)
\end{align*}

The binary confounder assigns each individual to exposure-seeking or exposure-averse behaviors. 
However, our test is valid, even when such confounding exists, as shown in Figure \ref{FIG: App sens conf}. $T_1$ is distributed as in the main analysis, and waning happens only through helped-to-doomed transitions. We fix $p_{E_{low}}=0.1$ and vary $p_{E_{high}} \in \{0.1, 0.3, 0.9\}$.
\begin{figure}[ht!]
    \centering
    \includegraphics[width=0.75\linewidth]{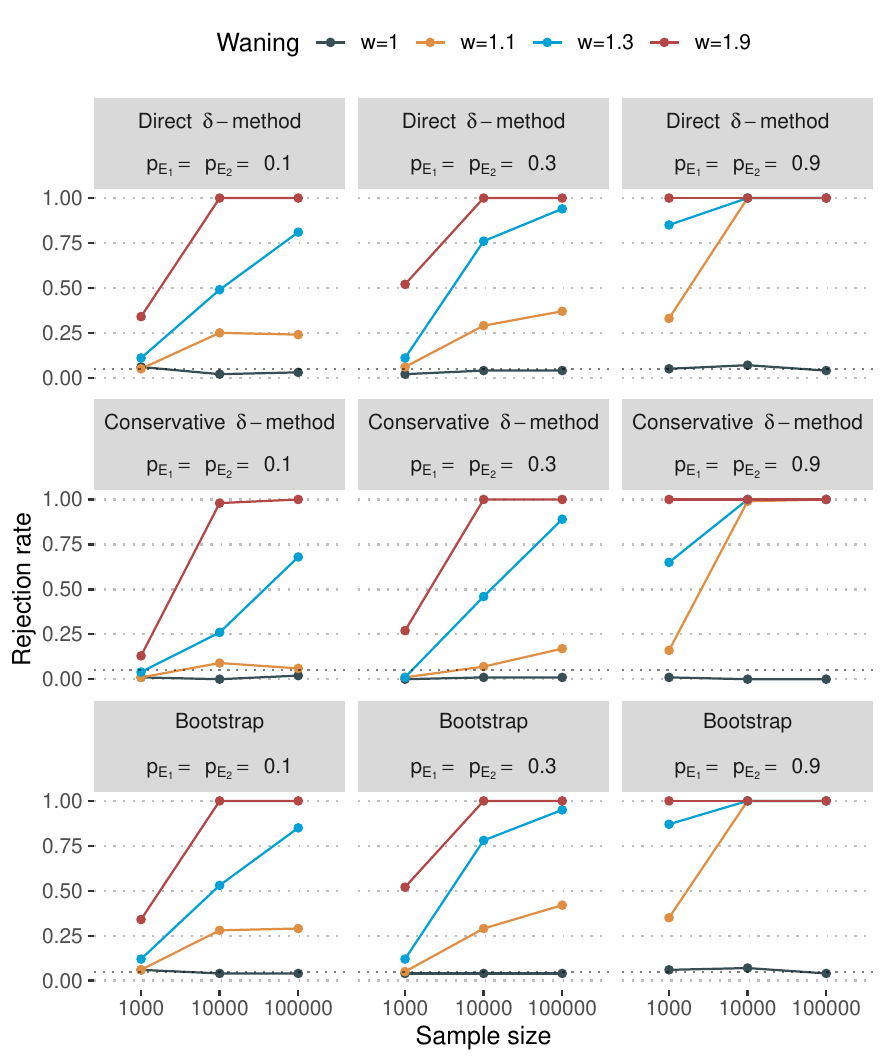}
    \caption{Rejection rates with confounded exposures.}
    \label{FIG: App sens conf}
\end{figure}

In all four modified scenarios, the power gain by using the direct $\delta$-method instead of the conservative approach is pronounced. The direct $\delta$-method has rejection rates approximately at the level of the test statistics, while in the case of the non-parametric bootstrap, larger deviations can happen. In general, we argue for the usage of the direct $\delta$-method, since it is feasible using summary data only, and has acceptable power, while maintaining the validity of the test.

\subsection{Determining transition rates}\label{APP_sub: wane spec}
In Appendix \ref{APP: Simulation} we simulate data, such that the distribution of the principal strata changes from time 1 to time 2. In particular, the ``challenge incidence ratio" should increase, that is, the incidence ratio should increase by a factor $w$ from time 1 to time 2 if we compare a population exposed at time 1 to a population that has been isolated first and then exposed at time 2. Consequently, the change in the principal strata, should be independent of all other variables.

We argued in Section \ref{SEC: Inference}, that the sharp null hypothesis is equivalent to the population level null hypothesis, under a set of conditions, one of which is sharp no waning of the placebo. That is, the individual level potential outcome will not change from time 1 to time 2 when intervening on treatment and exposure. Therefore, the only possible transitions are from ``doomed" to ``helped", ``helped" to ``doomed", ``harmed" to ``immune", and ``immune" to ``harmed". The waning of the vaccine, that is, the decrease in efficacy, corresponds to transitions from helped to doomed and immune to harmed (although, if for some, the vaccine efficacy increased, but for a larger proportion of the population it decreased, the vaccine still waned at a population level). For simplicity, in Appendix \ref{APP: Simulation} we only illustrated the results for the transition helped to doomed. Simulation results with alternative waning are available in Appendix \ref{APP: Sens}.

The incidence ratios in the challenge sense are
\begin{align*}
    &&IR_1^{\mathrm{challenge}}=\frac{p_{doomed_1}+p_{harmed_1}}{p_{doomed_1}+p_{helped_1}} &&IR_2^{\mathrm{challenge}}=\frac{p_{doomed_2}+p_{harmed_2}}{p_{doomed_2}+p_{helped_2}}.
\end{align*}
Suppose that the incidence rate increased by a factor of $w$ between the two time-points, due to the transitions helped to doomed and immune to harmed. In particular, suppose that a proportion $w^*_{helped}$ of the helped transitioned to doomed, a proportion $w^*_{immune}$ of the immune transitioned to harmed. Then the time 2 incidence ratio can be rewritten as 
\begin{align*}
    IR_2^{\mathrm{challenge}} &= w\times \frac{p_{doomed_1}+p_{harmed_1}}{p_{doomed_1}+p_{helped_1}}\\
    &= \frac{p_{doomed_1}+p_{harmed_1}+w^*_{helped} \times p_{helped_1}+w^*_{immune}\times p_{immune_1}}{p_{doomed_1}+p_{helped_1}}.
\end{align*}

Solving for $w^*_{helped}$, $w^*_{immune}$, we have 
\begin{equation}\label{EQ: transition rates}
        w^*_{helped}\times p_{helped_1}+w^*_{immune}\times p_{immune_1}= (w-1)\times (p_{doomed_1}+p_{harmed_1}).
\end{equation}

Thus, when the only transition happening is helped to doomed, 
$$
    w^*_{helped} = (w-1)\frac{p_{doomed_1}+p_{harmed_1}}{p_{helped_1}},
$$

and when the only transition is from immune to harmed,
$$
    w^*_{immune} = (w-1)\frac{p_{doomed_1}+p_{harmed_1}}{p_{immune_1}}.
$$

From Equation \eqref{EQ: transition rates} it is clear, that when $w=1$, that is, the sharp null hypothesis holds, there is no transition from either of the two possible principal strata.

\subsection{Interpretation of the IR, under the alternative}\label{APP: alternative}

Suppose that Hypothesis \ref{HYP: formal null} does not hold, that is, $\Delta Y_2^{a, e_1=0, e_2=1}$ and $\Delta Y_1^{a, e_1=1}$ can potentially differ. Under Assumptions \ref{ASS: Exp nec} and \ref{ASS: blind exp}, by identical arguments as in Appendix \ref{APP: Prop proof}, we can show that $IR_1 = 1-VE_1^{\mathrm{challenge}}$. Since $IR_1=HR_1$, our result agrees with the point identification claim of \citet{janvin_quantification_2024}.

By the same reasoning as in Appendix \ref{APP: Prop proof}, we can show that under Assumptions \ref{ASS: Exp nec}-\ref{ASS: exp exch}, \ref{ASS: Cons}-\ref{ASS: pos}
\begin{align*}
    E[\Delta Y_2|A=a] &=E[\Delta Y_2^{a}|E_1^{a}=0]P(E_1^{a}=0)\\
    &+E[\Delta Y_2^{a}| E_1^{a}=1]P(E_1^{a}=1),\\
\end{align*}
and that
$$
E[\Delta Y_2^{a}|E_1^{a}=0]P(E_1^{a}=0)= E[\Delta Y_2^{a, e_1=0, e_2=1}]P(E_2^{a, e_1=0}=1, E_1^{a}=0).
$$

However, under the alternative of the null hypothesis $E[\Delta Y_2^{a}|E_1^{a}=1]P(E_1^{a}=1)$ is not necessarily equal to $0$. In addition to the previously presented assumptions, suppose the following:

\begin{assumption}[Extended exposure exchangeability]\label{ASS: ext exp exch}
    \begin{align*}
        E_1^{a} &\independent \Delta Y_2^{a, e_1=1}|\Delta Y_1^{a, e_1=1}=0, E_2^{a, e_1=1}\\
        & \text{and} \qquad E_2^{a, e_1=1} \independent \Delta Y_2^{a, e_1=1, e_2}|\Delta Y_1^{a, e_1=1}=0\\
    \end{align*}
\end{assumption}
Assumption \ref{ASS: ext exp exch} generalizes Assumption \ref{ASS: exp exch} such that it holds under the intervention $e_1=1$. While previously, by Assumption \ref{ASS: Exp nec}, the intervention $e_1=0$ resulted in implicit conditioning on $\Delta Y_1^{a, e_1=0}=0$, since under $e_1=1$, the time 1 outcome is not deterministically 0, we have to condition on it explicitly. Both Assumptions \ref{ASS: exp exch} and \ref{ASS: ext exp exch} hold under the DAG proposed by \citet{janvin_quantification_2024}, reproduced in Figure \ref{FIG: DAG}.

Continuing with Equation \eqref{EQ: time1 exposed}
\begin{align*}
    &E[\Delta Y_2^{a, e_1=1}|E_1^{a}=1, \Delta Y_1^{a, e_1=1}=0, E_2^{a, e_1=1}=1]\\
    &\times P(E_2^{a, e_1=1}=1|E_1^{a}=1, \Delta Y_1^{a, e_1=1}=0) P(\Delta Y_1^{a, e_1=1}=0|E_1^{a}=1)P(E_1^{a}=1)\\
    &\overset{\mathrm{Assumption}\; \ref{ASS: ext exp exch}}{= }E[\Delta Y_2^{a, e_1=1}| \Delta Y_1^{a, e_1=1}=0, E_2^{a, e_1=1}=1]\\
    &\times P(E_2^{a, e_1=1}=1|E_1^{a}=1, \Delta Y_1^{a, e_1=1}=0) P(\Delta Y_1^{a, e_1=1}=0|E_1^{a}=1)P(E_1^{a}=1)\\
    &\overset{\mathrm{Assumption}\; \ref{ASS: Cons}}{= }E[\Delta Y_2^{a, e_1=1, e_2=1}| \Delta Y_1^{a, e_1=1}=0, E_2^{a, e_1=1}=1]\\
    &\times P(E_2^{a, e_1=1}=1|E_1^{a}=1, \Delta Y_1^{a, e_1=1}=0) P(\Delta Y_1^{a, e_1=1}=0|E_1^{a}=1)P(E_1^{a}=1)\\
    &\overset{\mathrm{Assumption}\; \ref{ASS: ext exp exch}}{= }E[\Delta Y_2^{a, e_1=1, e_2=1}|\Delta Y_1^{a, e_1=1}=0] P(E_2^{a, e_1=1}=1|E_1^{a}=1, \Delta Y_1^{a, e_1=1}=0)\\
     &\times P(\Delta Y_1^{a, e_1=1}=0|E_1^{a}=1)P(E_1^{a}=1)\\
     &\overset{\mathrm{Assumption}\; \ref{ASS: exp exch}}{= }E[\Delta Y_2^{a, e_1=1, e_2=1}|\Delta Y_1^{a, e_1=1}=0] P(\Delta Y_1^{a, e_1=1}=0)\\
     &\times P(E_2^{a, e_1=1}=1|E_1^{a}=1, \Delta Y_1^{a, e_1=1}=0) P(E_1^{a}=1)\\
     &\overset{\mathrm{Assumption}\; \ref{ASS: Cons}}{= }E[\Delta Y_2^{a, e_1=1, \delta y_1=0, e_2=1}|\Delta Y_1^{a, e_1=1}=0] P(\Delta Y_1^{a, e_1=1}=0)\\
     &\times P(E_2^{a, e_1=1}=1|E_1^{a}=1, \Delta Y_1^{a, e_1=1}=0) P(E_1^{a}=1)\\
     &\overset{\mathrm{Assumption}\; \ref{ASS: exp eff restriction}}{= }E[\Delta Y_2^{a, e_1=0, e_2=1}|\Delta Y_1^{a, e_1=1}=0] P(\Delta Y_1^{a, e_1=1}=0)\\
     &\times P(E_2^{a, e_1=1}=1|E_1^{a}=1, \Delta Y_1^{a, e_1=1}=0) P(E_1^{a}=1)\\
     &=P(\Delta Y_2^{a, e_1=0, e_2=1}=1, \Delta Y_1^{a, e_1=1}=0)\\
     &\times P(E_2^{a, e_1=1}=1|E_1^{a}=1, \Delta Y_1^{a, e_1=1}=0) P(E_1^{a}=1)
\end{align*}

Under the null hypothesis, the joint probability,
$$
P(\Delta Y_2^{a, e_1=0, e_2=1}=1, \Delta Y_1^{a, e_1=1}=0)
$$
is equal to $0$; however, under the alternative, not necessarily.

Therefore, 
\begin{equation}\label{EQ: Incidence time 2}
    \begin{aligned}
        E[\Delta Y_2|A=a]&= P(\Delta Y_2^{a, e_1=0, e_2=1}=1) P(E_2^{a, e_1=0}=1|E_1^{a}=0, \Delta Y_1^{a, e_1=0}=0) P(E_1^{a}=0)\\
        &+P(\Delta Y_2^{a, e_1=0, e_2=1}=1, \Delta Y_1^{a, e_1=1}=0) P(E_2^{a, e_1=1}=1|E_1^{a}=1, \Delta Y_1^{a, e_1=1}=0) P(E_1^{a}=1)
    \end{aligned}
\end{equation}
and 
$$
    IR_2 = \frac{P(\Delta Y_2^{a=1, e_1=0, e_2=1}=1)p^{2,a=1}_0p^1_0+P(\Delta Y_2^{a=1, e_1=0, e_2=1}=1, \Delta Y_1^{a=1, e_1=1}=0)p^{2,a=1}_1p^1_1}{P(\Delta Y_2^{a=0, e_1=0, e_2=1}=1)p^{2,a=0}_0p^1_0+P(\Delta Y_2^{a=0, e_1=0, e_2=1}=1, \Delta Y_1^{a=0, e_1=1}=0)p^{2,a=0}_1p^1_1},
$$
where $p^{2,a}_{e_1}= P(E_2^{a, e_1}=1|E_1^{a}=e_1, \Delta Y_1^{a, e_1}=0)$ and $p^1_{e_1}=P(E_1^{a}=e_1)\overset{\mathrm{Assumption}\; \ref{ASS: blind exp}}{=}P(E_1=e_1)$.

Depending on the exposure probabilities, $P(\Delta Y_2^{a, e_1=0, e_2=1}=1)$, and $P(\Delta Y_2^{a, e_1=0, e_2=1}=1, \Delta Y_1^{a, e_1=1}=0)$, $IR_2$ can be either below or above the $VE_2^{\mathrm{challenge}}$, thus $IR_2$ is uninformative about the extent of the waning. 

\begin{assumption}[Sharp no waning of the placebo]\label{ASS: sharp no a0 waning}
    $$
        \Delta Y_2^{a=0, e_1=0, e_2=1} = \Delta Y_1^{a, e_1=1} \quad \text{w.p } 1.
    $$
\end{assumption}

Assumption \ref{ASS: sharp no a0 waning} was used to establish equivalence between individual and population waning in Remark \ref{REM: waning and sharp}. Under the alternative and Assumption \ref{ASS: sharp no a0 waning}, we allow for $\Delta Y_2^{a=1, e_1=0, e_2=1} \neq \Delta Y_1^{a=1, e_1=1} $ with a non-zero probability.

By Assumption \ref{ASS: sharp no a0 waning}
\begin{align*}
    IR_2&= \frac{P(\Delta Y_2^{a=1, e_1=0, e_2=1}=1)p^{2,a=1}_0p^1_0+P(\Delta Y_2^{a=1, e_1=0, e_2=1}=1, \Delta Y_1^{a=1, e_1=1}=0)p^{2,a=1}_1p^1_1}{P(\Delta Y_2^{a=0, e_1=0, e_2=1}=1)p^{2,a=0}_0p^1_0}\\
    &=\frac{P(\Delta Y_2^{a=1, e_1=0, e_2=1}=1)p^{2,a=1}_0p^1_0}{P(\Delta Y_2^{a=0, e_1=0, e_2=1}=1)p^{2,a=0}_0p^1_0}+\frac{P(\Delta Y_2^{a=1, e_1=0, e_2=1}=1, \Delta Y_1^{a=1, e_1=1}=0)p^{2,a=1}_1p^1_1}{P(\Delta Y_2^{a=0, e_1=0, e_2=1}=1)p^{2,a=0}_0p^1_0}\\
    &=1-VE_2^{\mathrm{challenge}}+\frac{P(\Delta Y_2^{a=1, e_1=0, e_2=1}=1, \Delta Y_1^{a=1, e_1=1}=0)p^{2,a=1}_1p^1_1}{P(\Delta Y_2^{a=0, e_1=0, e_2=1}=1)p^{2,a=0}_0p^1_0}.
\end{align*}

In other words, under Assumption \ref{ASS: sharp no a0 waning} and the alternative, the second timepoint incidence ratio overestimates $1-VE_2^{\mathrm{challenge}}$. The size of the overestimation depends on the exposure probabilities and the size of the principal strata, which corresponds to individuals who transitioned from immune to harmed or from helped to doomed. Those who experienced the improvement of the vaccine over time (doomed to helped and harmed to immune) do not add to the bias, since even if they remained event-free in time 1, in the second interval they do not develop the outcome under treatment, that is the same outcome that they would have experienced, had they been isolated at time 1.

Due to the sign of the bias, when $IR_1<IR_2$, that is, based on the incidence ratios, the vaccine seemingly wanes, $VE_1^{\mathrm{challenge}}<VE_2^{\mathrm{challenge}}$ can still hold. The qualitative difference between the incidence ratio and the challenge effects, under Assumption \ref{ASS: sharp no a0 waning}, can occur if
\begin{align*}
    P(\Delta Y_2^{a=1, e_1=0, e_2=1}=0, \Delta Y_1^{a=1, e_1}=1) > P(\Delta Y_2^{a=1, e_1=0, e_2=1}=1, \Delta Y_1^{a=1, e_1}=0),
\end{align*}
that is, the probability of the effect of the individual vaccine improvement is greater than the probability of individual waning, and
\begin{align*}
    &P(\Delta Y_2^{a=1, e_1=0, e_2=1}=1, \Delta Y_1^{a=1, e_1}=0)(p_0^{2, a=1}p_0^1+p_1^{2, a=1}p_1^1)\\
    &> P(\Delta Y_2^{a=1, e_1=0, e_2=1}=0, \Delta Y_1^{a=1, e_1}=1)p_0^{2, a=1}p_0^1.
\end{align*}

Assuming independence between the exposures, given survival ($E_1^{a} \independent E_2^{a, e_1}|\Delta Y_1^{a, e_1}=0$, similar to Equation 7 of \citet{janvin_quantification_2024}), $p_1^{2, a}=p_0^{2,a}$ follows and the second condition simplifies to 
\begin{align*}
    &P(\Delta Y_2^{a=1, e_1=0, e_2=1}=1, \Delta Y_1^{a=1, e_1}=0) > P(\Delta Y_2^{a=1, e_1=0, e_2=1}=0, \Delta Y_1^{a=1, e_1}=1)p_0^1.
\end{align*}
Therefore, if the extent of waning is substantial compared to vaccine improvement and the infectious pressure at time 1 is high, waning could be mistakenly claimed based on the incidence ratios, even though the challenge effect has increased over time.

If $IR_2<IR_1$, then due to the positive bias of $IR_2$, $VE_2^{\mathrm{challenge}}>VE_1^{\mathrm{challenge}}$, that is, if the vaccine seemingly improved based on the incidence ratios, indeed the vaccine improved in the challenge effect sense. However, the size of this improvement is overestimated by the incidence ratios, with the magnitude of the overestimation depending on the probability of vaccine waning and the exposure probabilities.

\subsection{Sensitivity analysis on the challenge effects}\label{APP_sub: Challenge sens}

From Equation \eqref{EQ: Incidence time 2} it follows that 
$$
P(\Delta Y_2^{a, e_1=0, e_2=1}=1)p_0^{2, a}p_0^{1}=E[\Delta Y_2|A=a]- P(\Delta Y_2^{a, e_1=0, e_2=1}=1, \Delta Y_1^{a, e_1=1}=0)p_1^{2, a}p_1^{1}.
$$

Under Assumption \ref{ASS: blind exp} and \ref{ASS: Cons}, $p_0^{2, a=1}p_0^{1}=p_0^{2, a=0}p_0^{1}$, thus

\begin{equation}\label{EQ: challenge 2}
    1-VE_2^{challenge}=\frac{E[\Delta Y_2|A=1]- P(\Delta Y_2^{a=1, e_1=0, e_2=1}=1, \Delta Y_1^{a=1, e_1=1}=0)p_1^{2, a=1}p_1^{1}}{E[\Delta Y_2|A=0]- P(\Delta Y_2^{a=0, e_1=0, e_2=1}=1, \Delta Y_1^{a=0, e_1=1}=0)p_1^{2, a=0}p_1^{1}}.
\end{equation}

Using information from the observed data and Assumption \ref{ASS: sharp no a0 waning}, we can bound some of the expressions in Equation \eqref{EQ: challenge 2}. As
\begin{align*}
    E[Y_2|A=a]&=\big(P(\Delta Y_2^{a, e_1=0, e_2=1}=1, \Delta Y_1^{a, e_1=1}=0)+P(\Delta Y_2^{a, e_1=0, e_2=1}=1, \Delta Y_1^{a, e_1=1}=1)\big)p^{2}_0p^1_0\\
    &+P(\Delta Y_2^{a, e_1=0, e_2=1}=1, \Delta Y_1^{a, e_1=1}=0)p^{2}_1p^1_1,
\end{align*}
thus 
\begin{equation}\label{EQ: waning strata restriction}
    P(\Delta Y_2^{a, e_1=0, e_2=1}=1, \Delta Y_1^{a, e_1=1}=0) \leq \frac{E[\Delta Y_2|A=a]}{p^{2}_0p^1_0+p^{2}_1p^1_1}.
\end{equation}

Under Assumption \ref{ASS: sharp no a0 waning} $P(\Delta Y_2^{a=0, e_1=0, e_2=1}=1, \Delta Y_1^{a=0, e_1=1}=0)=0$, thus
\begin{align*}
    E[\Delta Y_1|A=0]&=P(\Delta Y_2^{a=0, e_1=0, e_2=1}=1, \Delta Y_1^{a=0, e_1=1}=1) p_1^1\\
&\text{and}\\
E[\Delta Y_2|A=0]&=P(\Delta Y_2^{a=0, e_1=0, e_2=1}=1, \Delta Y_1^{a=0, e_1=1}=1)p_0^2(1-p_1^1),
\end{align*}
implying that the equality 
\begin{equation}\label{EQ: exp prob restriction}
    \frac{E[\Delta Y_1|A=0]}{E[\Delta Y_2|A=0]} = \frac{p_1^1}{p_0^2(1-p_1^1)}
\end{equation}
must hold.

Since under Assumption \ref{ASS: sharp no a0 waning}
\begin{align*}
    VE_2^{\mathrm{challenge}}&=1-IR_2+\frac{P(\Delta Y_2^{a=1, e_1=0, e_2=1}=1, \Delta Y_1^{a=1, e_1=1}=0)p_1^2p_1^1}{E[\Delta Y_2|A=0]},
\end{align*}
$VE_2^{\mathrm{challenge}}$ is maximized when 
$P(\Delta Y_2^{a=1, e_1=0, e_2=1}=1, \Delta Y_1^{a=1, e_1=1}=0)p_1^2p_1^1- E[\Delta Y_2|A=1]$ reaches its maximum. By Equation \eqref{EQ: waning strata restriction}, the maximum is reached at $P(\Delta Y_2^{a=1, e_1=0, e_2=1}=1, \Delta Y_1^{a=1, e_1=1}=0) = \frac{E[\Delta Y_2|A=1]}{p^{2}_0p^1_0+p^{2}_1p^1_1}$, thus it is equal to 
$$
E[\Delta Y_2|A=1]\frac{p_1^2p_1^1-p^{2}_0p^1_0-p^{2}_1p^1_1}{p^{2}_0p^1_0+p^{2}_1p^1_1} = E[\Delta Y_2|A=1]\frac{-p^{2}_0p^1_0}{p^{2}_0p^1_0+p^{2}_1p^1_1}.
$$

By Equation \eqref{EQ: exp prob restriction}, 
\begin{align*}
    p_1^1 &= \frac{E[\Delta Y_1|A=0]p_0^2}{E[\Delta Y_2|A=0]+E[\Delta Y_1|A=0]p_0^2}
    &\text{and} && p_0^1 &= \frac{E[\Delta Y_2|A=0]}{E[\Delta Y_2|A=0]+E[\Delta Y_1|A=0]p_0^2},
\end{align*}
thus, the maximum simplifies to 
$$
E[\Delta Y_2|A=1]\frac{-p_0^2 E[\Delta Y_2|A=0]}{p_0^2 E[\Delta Y_2|A=0] + p_1^{2}E[\Delta Y_1|A=0] p_0^2} = E[\Delta Y_2|A=1]\frac{- E[\Delta Y_2|A=0]}{E[\Delta Y_2|A=0] + p_1^{2}E[\Delta Y_1|A=0]}.
$$

Since all expectations are non-negative, the fraction increases in $p_1^2$; therefore, it reaches its maximum at $p_1^2=1$. Then,
\begin{align*}
    VE_2^{\mathrm{challenge}} &\leq 1+ \frac{E[\Delta Y_2|A=1]}{E[\Delta Y_2|A=0]}\frac{- E[\Delta Y_2|A=0]}{E[\Delta Y_2|A=0] + p_1^{2}E[\Delta Y_1|A=0]} \Big |_{p_1^2=1} \\
    &=1- \frac{E[Y_2|A=1]}{E[Y_1|A=0]+E[Y_2|A=0]},
\end{align*}
which is the same upper bound on the challenge effect as derived by \citet{janvin_quantification_2024}. However, if data are available from a standard RCT and there is knowledge about the probability 
$P(E_2^{a, e_1=1}=1| E_1^{a}=1, \Delta Y_1^{a, e_1=1}=0)$, for example, from contact tracing studies, the upper bound can be considerably improved.

To illustrate the result, we took the observed quantities $E[Y_1|A=0], E[Y_1|A=1], E[Y_2|A=0]$, and $E[Y_2|A=1]$ from the BNT162b2 trial \cite{thomas_safety_2021}, and plotted the upper bound under the changing values of $p_1^2$ in Figure \ref{FIG: Upper bound}

\begin{figure}[!htbp]
    \centering
    \includegraphics[width=1\linewidth]{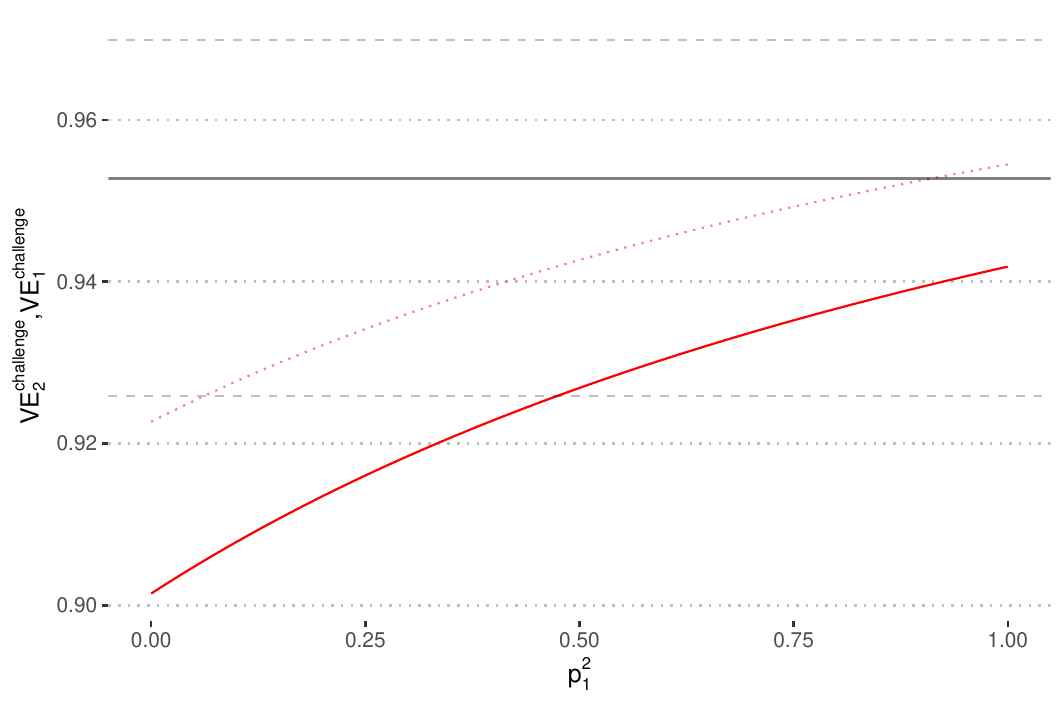}
    \caption{The changing value on the upper bound of $VE_2^{\mathrm{challenge}}$ (solid red line), for some $p_1^2$, with the corresponding one-sides $95\%$ confidence interval (dashed red). $VE_1^{\mathrm{challenge}}$ (solid black), is point-identified, regardless of the value of $p_1^2$ or other exposure probabilities. A dashed black line represents the corresponding $95\%$ two-sided confidence interval. The values at $p_1^2=1$ are equal to the findings of \citet{janvin_quantification_2024}.}
    \label{FIG: Upper bound}
\end{figure}

\subsection{On the Assumptions of \citet{janvin_quantification_2024}}\label{APP: on Ass-s}
Similar to \citet{janvin_quantification_2024}, because we consider data from a randomized controlled trial, we take the following three conventional assumptions for granted:
\begin{assumption}[Consistency] \label{ASS: Cons}
Suppose that the interventions on the treatment $A$, exposures $E_1$, $E_2$, and the first timepoint survival $\Delta Y_1=0$ are well-defined. That is
\begin{enumerate}[label = (\roman*)]
    \item if $A=a$ then $E_1=E_1^{a}, \Delta Y_1=\Delta Y_1^{a}, E_2= E_2^{a}, \Delta Y_2=\Delta Y_2^{a}$,
    \item if $E_1^{a}=e_1$ then $\Delta Y_1^{a}=\Delta Y_1^{a, e_1}, E_2^{a}=E_2^{a, e_1}, \Delta Y_2^{a}=\Delta Y_2^{a, e_1}$,
    \item if $\Delta Y_1^{a, e_1}=0$ then $E_2^{a, e_1}=E_2^{a, e_1, \delta y_1=0}, \Delta Y_2^{a, e_1}=\Delta Y_2^{a, e_1, \delta y_1=0}$,
    \item if $E_2^{a, e_1, \delta y_1=0}=e_2$ then $\Delta Y_2^{a, e_1, \delta y_1=0}=\Delta Y_2^{a, e_1, \delta y_1=0, e_2}$,
    \item if $E_2^{a, e_1=0}=e_2$ then $\Delta Y_2^{a, e_1=0}=\Delta Y_2^{a, e_1=0, e_2}$,
\end{enumerate}
    for all $a, e_1, e_2 \in \{0,1\}$.
\end{assumption}

\begin{assumption}[Treatment exchangeability]\label{ASS: Trt exch}
    $$
        E_1^{a}, \Delta Y_1^{a}, \Delta Y_1^{a, e_1}, E_2^{a, e_1}, E_2^{a, e_1 \delta y_1=0}, \Delta Y_2^{a, e_1,  \delta y_1=0, e_2} \independent A
    $$
    for all $a, e_1, e_2 \in \{0,1\}$.
\end{assumption}

\begin{assumption}[Positivity]\label{ASS: pos}
    $$P(A=a)>0 \; \text{ for all } a \in \{0,1\}.$$
\end{assumption}

Assumptions \ref{ASS: Cons}-\ref{ASS: pos} are close to the RCT assumptions set by \citet{janvin_quantification_2024}. The most important distinction is that we suppose it is possible to intervene on $\Delta Y_1$, with a well-defined intervention; thus, consistency holds if $\Delta Y_1=0$. Accordingly, we state the treatment exchangeability assumption for potential outcomes under the intervention $\delta y_1=0$.

Assumption \ref{ASS: Exp nec} states that in order to develop the outcome at a given timepoint, exposure at the same timepoint should have happened. The version stated in the main document is stronger than the one proposed by \citet{janvin_quantification_2024}: Assumption \ref{ASS: Exp nec} needs to hold even if there was no intervention at time 1 isolating the individuals. That is, there can be no causal effect of time 1 exposure at the time 2 outcome, that is not mediated through $\Delta Y_1$ or $E_2$, which is also a sufficient condition for Assumption \ref{ASS: exp eff restriction}.

While Assumption \ref{ASS: Trt exch} already imposed structural constraints on the relationship of the variables, we can take those for granted by the randomized design. Assumptions \ref{ASS: exp eff restriction}-\ref{ASS: exp exch} further restrict the dependencies between the variables. \citet{janvin_quantification_2024} argues for a plausible structure of the observed data, depicted in Figure \ref{FIG: DAG}.
The casual DAG satisfies Assumption \ref{ASS: exp eff restriction}-\ref{ASS: exp exch} and \ref{ASS: ext exp exch}.

\begin{figure}
    \centering
            \centering
        \begin{tikzpicture}
            \tikzset{line width=1.5pt, outer sep=0pt,
            ell/.style={draw,fill=white, inner sep=2pt,
            line width=1.5pt},
            swig vsplit={gap=5pt,
            inner line width right=0.5pt},
            swig hsplit={gap=5pt}
            };
                \node[name=UY, ell, shape=ellipse] at (4.5,6){$U_Y$};
                \node[name=DY1, ell, shape=ellipse] at (3,4){$\Delta Y_1$};
                \node[name=DY2, ell, shape=ellipse] at (6,4){$\Delta Y_2$};
                \node[name=A, ell, shape=ellipse] at (0,4){$A$};
                \node[name=E1, ell, shape=ellipse] at (3,2) {$E_1$};
                \node[name=E2, ell, shape=ellipse] at (6,2) {$E_2$};
                \node[name=UE, ell, shape=ellipse] at (4.5,0) {$U_E$};
            \begin{scope}[>={Stealth[black]},
              every node/.style={fill=white,circle},
              every edge/.style={draw=gray,very thick}]
                \path[->] (UY) edge (DY1);
                \path[->] (UY) edge (DY2);

                \path[->] (DY1) edge (DY2);
                \path[->] (DY1) edge (E2);

                \path[->] (A) edge (DY1);
                \path[->] (A) edge[bend right=30] (DY2);
                
                \path[->] (E1) edge (DY1);
                \path[->] (E2) edge (DY2);
                
                \path[->] (UE) edge (E1);
                \path[->] (UE) edge (E2);
                \end{scope}
        \end{tikzpicture}
        \caption{Directed acyclic graph illustrating the observed data structure.}
        \label{FIG: DAG}
\end{figure}
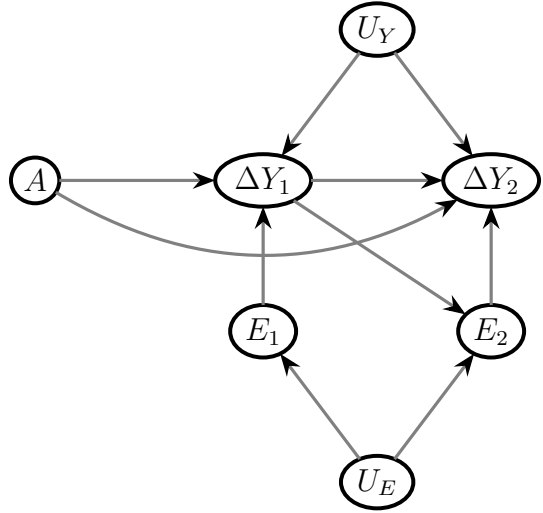

Assumption \ref{ASS: exp eff restriction} is not stated explicitly by \citet{janvin_quantification_2024}, but, the same argument is used for their Assumption 4: no causal path is expected to exist from $E_1$ to $\Delta Y_2$ that is not through $\Delta Y_1$ (we also allow mediation through $E_2$). However, if there exists such a path under an intervention $\delta y_1=0$, Assumption \ref{ASS: exp eff restriction} is violated, while Assumption 4 of \citet{janvin_quantification_2024} can still hold, even though this occurs in a contrived scenario.

Assumption \ref{ASS: exp eff restriction} facilitates our discussion and it is conceivable in many scenarios, for example antiretroviral post-exposure prophylaxis (PEP) treatment following exposure to HIV, or defining exposure intervention as spraying infectious particles of the respiratory disease towards the patient, however ensuring survival by an intervention instructing them to wear medical face shields. However, Assumption \ref{ASS: exp eff restriction} is not essential for our results to hold, and we can avoid using potential outcomes under interventions on survival.

\begin{assumption}\label{ASS: mon}
    $$
    E[\Delta Y_2^{a, e_1=1, e_2=1}|\Delta Y_1^{a, e_1=1}=0]=E[\Delta Y_2^{a, e_1=0, e_2=1}|\Delta Y_1^{a, e_1=1}=0]
    $$
\end{assumption}

Assumption \ref{ASS: mon} posits that, among individuals who survived time 1, despite exposure, the expected event rate at time 2 is the same following isolation as it is following exposure. Although Assumption \ref{ASS: mon} is cross-world, it is plausible when $E_1$ has no direct effect on $\Delta Y_2$. In that case the effect of $E_1$ is entirely mediated through time 1 survival. Accordingly in models such as NPSEM-IE-s \citep{pearl_causality_2009}, conditioning on survival makes the expected time 2 outcome invariant to interventions on $E_1$.

Under the null hypothesis, $E[\Delta Y_2^{a, e_1=0, e_2=1}|\Delta Y_1^{a, e_1=1}=0]=0$, thus under Assumption \ref{ASS: mon} $E[\Delta Y_2^{a, e_1=1, e_2=1}|\Delta Y_1^{a, e_1=1}=0]=0$, thus Equation \eqref{EQ: cross-trick} equals to 0, which shows that $E[\Delta Y_2^{a}|E_1^{a}=1]P(E_1^{a}=1)=0$, and proves our result.

Assumption \ref{ASS: blind exp} is equivalent to Assumption 3 of \citet{janvin_quantification_2024}.

Finally, under Assumption \ref{ASS: Trt exch}, Assumption \ref{ASS: exp exch} is equivalent to Assumption 4 of \citet{janvin_quantification_2024}. 

Therefore, to test the null hypothesis, Assumption \ref{ASS: Exp nec} and \ref{ASS: exp eff restriction} slightly restrict their counterparts from \citet{janvin_quantification_2024}. However, no other independencies are imposed that would not hold under the DAG proposed to represent a plausible data-generating mechanism. While all of the above discussed assumptions hold under the DAG in Figure \ref{FIG: DAG}, they would hold even if $E_1$ had a direct effect on $E_2$, thus in this sense, our conditions are more general than the ones imposed by \citet{janvin_quantification_2024}.

\subsection{Hazard-ratio based waning under the null hypothesis}\label{APP: Hazard}
Suppose that, similarly to Appendix \ref{APP: Simulation}, the variables are distributed according to the following data-generating mechanism:
\begin{align}
    A &\sim P_A\\
    T & \sim \mathrm{Categorical}\{p_{doomed}, p_{helped}, p_{harmed}, p_{immune}\}\\
    E_1 &\sim P_{E_1}\\
    \Delta Y_1 & = E_1 \times \big(I(T_1=1)+(1-A)\times I(T=2)+ A \times I(T=3) \big)\\
    E_2 &= (1-\Delta Y_1) \times P_{E_2} \label{EQ: Time 2 exp}\\
    \Delta Y_2 & = E_2 \times \big(I(T_1=1)+(1-A)\times I(T=2)+ A \times I(T=3) \big)
\end{align}
For simplicity, consider $P_A\sim \mathrm{Bernoulli}(0.5)$, and $P_{E_1} \sim \mathrm{Bernoulli}(p_{E_1})$, $P_{E_2}\sim \mathrm{Bernoulli}(p_{E_2})$. Since at randomization all individuals are event-free, it follows that 
$$
IR_1=HR_1\overset{H_0}{=}\frac{P\big(T \in \{11,16\} \big)}{P\big(T \in \{6,16\} \big)}=\frac{p_{harmed}+p_{doomed}}{p_{helped}+p_{doomed}},
$$
as shown in Appendices \ref{APP: PS general} and \ref{APP: Prop proof}.

However, at time 2, the discrete hazard can be rewritten as 
\begin{align*}
    E[\Delta Y_2|\Delta Y_1 =0, A=a] &= E[\Delta Y_2|\Delta Y_1=0, A=a, E_1=1]P(E_1=1|\Delta Y_1 =0, A=a)\\
    &+E[\Delta Y_2|\Delta Y_1=0, A=a, E_1=0]P(E_1=0|\Delta Y_1 =0, A=a)\\
    &\overset{\mathrm{Assumption}\; \ref{ASS: Cons}-\ref{ASS: pos}}{=} E[\Delta Y_2^{a}|\Delta Y_1^{a}=0, E_1^{a}=1]P(E_1^{a}=1|\Delta Y_1^{a} =0)\\
    &+E[\Delta Y_2^{a}|\Delta Y_1^{a}=0, E_1^{a}=0]P(E_1^{a}=0|\Delta Y_1^{a} =0)\\
    &\overset{\mathrm{Assumption}\; \ref{ASS: Exp nec}}{=}E[\Delta Y_2^{a}|\Delta Y_1^{a}=0, E_1^{a}=1, E_2^{a}=1]\\
    &\times P(E_2^{a}=1|\Delta Y_1^{a}=0, E_1^{a}=1) P(E_1^{a}=1|\Delta Y_1^{a} =0)\\
    &+E[\Delta Y_2^{a}|\Delta Y_1^{a}=0, E_1^{a}=0, E_2^{a}=1]\\
    &\times P(E_2^{a}=1|\Delta Y_1^{a}=0, E_1^{a}=0) P(E_1^{a}=0|\Delta Y_1^{a} =0)\\
    &\overset{\mathrm{Assumption}\; \ref{ASS: Cons}}{=}E[\Delta Y_2^{a, e_1=1, e_2=1}|\Delta Y_1^{a, e_1=1}=0, E_1^{a}=1, E_2^{a, e_1=1}=1]\\
    &\times P(E_2^{a}=1|\Delta Y_1^{a}=0, E_1^{a}=1) P(E_1^{a}=1|\Delta Y_1^{a} =0)\\
    &+E[\Delta Y_2^{a, e_1=0, e_2=1}|\Delta Y_1^{a, e_1=0}=0, E_1^{a}=0, E_2^{a, e_1=0}=1]\\
    &\times P(E_2^{a}=1|\Delta Y_1^{a}=0, E_1^{a}=0) P(E_1^{a}=0|\Delta Y_1^{a} =0)\\
    &\overset{\mathrm{Assumption}\; \ref{ASS: exp exch}, \, \ref{ASS: ext exp exch}}{=}E[\Delta Y_2^{a, e_1=1, e_2=1}|\Delta Y_1^{a, e_1=1}=0]\\
    &\times P(E_2^{a}=1|\Delta Y_1^{a}=0, E_1^{a}=1) P(E_1^{a}=1|\Delta Y_1^{a} =0)\\
    &+E[\Delta Y_2^{a, e_1=0, e_2=1}|\Delta Y_1^{a, e_1=0}=0]\\
    &\times P(E_2^{a}=1|\Delta Y_1^{a}=0, E_1^{a}=0) P(E_1^{a}=0|\Delta Y_1^{a} =0).
\end{align*}

Under Assumption \ref{ASS: exp eff restriction} and the null hypothesis, 
$$
    E[\Delta Y_2^{a, e_1=1, e_2=1}|\Delta Y_1^{a, e_1=1}=0]=0.
$$

The conditional probability of no exposure can be rewritten as 
\begin{align*}
    P(E_1^{a}=0|\Delta Y_1^{a}=0)&=\frac{P(\Delta Y_1^{a}=0|E_1^{a}=0)P(E_1^{a}=0)}{P(\Delta Y_1^{a}=0|E_1^{a}=1)P(E_1^{a}=1)+P(\Delta Y_1^{a}=0|E_1^{a}=0)P(E_1^{a}=0)}\\
    &\overset{\mathrm{Assumption}\; \ref{ASS: Exp nec}}{=}\frac{P(E_1^{a}=0)}{P(\Delta Y_1^{a}=0|E_1^{a}=1)P(E_1^{a}=1)+P(E_1^{a}=0)}\\
    &\overset{\mathrm{Assumption}\;\ref{ASS: Cons},\, \ref{ASS: exp exch}}{= }\frac{P(E_1^{a}=0)}{P(\Delta Y_1^{a, e_1=1}=0)P(E_1^{a}=1)+P(E_1^{a}=0)}.\\
\end{align*}

Similarly, under Assumption \ref{ASS: Exp nec}
$$
E[\Delta Y_2^{a, e_1=0, e_2=1}|\Delta Y_1^{a, e_1=0}=0] = P(\Delta Y_2^{a, e_1=0, e_2=1}=1, \Delta Y_1^{a, e_1=0}=0)=P(\Delta Y_2^{a, e_1=0, e_2=1}=1),
$$
where we used that $\Delta Y_1=1 \implies\Delta Y_2=0$

Therefore,
\begin{align*}
    E[\Delta Y_2|\Delta Y_1 =0, A=a] &=P(\Delta Y_2^{a, e_1=0, e_2=1}=1) P(E_2^{a}=1|\Delta Y_1^{a}=0, E_1^{a}=0)\\
    &\times  \frac{P(E_1^{a}=0)}{P(\Delta Y_1^{a, e_1=1}=0)P(E_1^{a}=1)+P(E_1^{a}=0)}.
\end{align*}

Given the data-generating mechanism, in particular the fact that there is no confounding between $E_1$ and $E_2$, we have $E_2 \independent (A, E_1)|\Delta Y_1=0$. Thus
\begin{align*}
    E[\Delta Y_2|\Delta Y_1=0, A=a]&=P(\Delta Y_2^{a, e_1=0, e_2=1}=1) P(E_2^{a}=1|\Delta Y_1^{a}=0)\\
    &\times \frac{P(E_1^{a}=0)}{P(\Delta Y_1^{a, e_1=1}=0)P(E_1^{a}=1)+P(E_1^{a}=0)}.
\end{align*}

Moreover, in the  proposed data-generating mechanism, 
$$
     P(E_2^{a=1}=1|\Delta Y_1^{a=1}=0)= P(E_2^{a=0}=1|\Delta Y_1^{a=0}=0),
$$
see Equation \eqref{EQ: Time 2 exp}. Thus, the hazard ratio at time 2, is equal to

\begin{align*}
    HR_2=\frac{P(\Delta Y_1^{a=1, e_1=1}=1)}{P(\Delta Y_1^{a=0, e_1=1}=1)}\frac{P(\Delta Y_1^{a=0, e_1=1}=0)P(E_1=1)+P(E_1=0)}{P(\Delta Y_1^{a=1, e_1=1}=0)P(E_1=1)+P(E_1=0)},
\end{align*}

where we used Assumption \ref{ASS: blind exp} to cancel out $P(E_1^{a}=1)$, and the fact that under the null hypothesis $P(\Delta Y_2^{a, e_1=0, e_2=1}=1)=P(\Delta Y_1^{a, e_1=1}=1)$.
Therefore, the hazard ratio does not only depend on the probability of the outcome given exposure, but also on the  probability of the time 1 exposure.
Even though the principal strata did not change over time, the time 2 hazard ratio is different from the time 1 hazard ratio by a factor of $\frac{P(\Delta Y_1^{a=0, e_1=1}=0)P(E_1=1)+P(E_1=0)}{P(\Delta Y_1^{a=1, e_1=1}=0)P(E_1=1)+P(E_1=0)}$. Thus, even if the first interval exposure increases, it is not guaranteed, that the ratio of the hazard ratios will increase or decrease monotonically, without knowledge of the distribution of the principal strata. See simulation results on testing the sharp null hypothesis of no waning (Hypothesis \ref{HYP: formal null}) using the equality of the hazard ratios in Appendix \ref{APP: Simulation} of the main document.

\subsection{Variance estimation}\label{APP: Variance}
To construct valid statistical tests for $\frac{\widehat{IR}_1}{\widehat{IR}_2}$, the limiting distribution of its estimator must be known. 

Denote with $\Delta Y_{i} \in \{0,1,2\}$ as the multivariate outcome, corresponding to the outcome of interest, not happening, happening at time 1, or time 2, respectively, and denote with $A_i \in \{0,1\}$ the treatment assignment of individual $i \in \{1, \dots, n\}$. Then, the estimators can be defined as
$$
\widehat{IR}_k :=\frac{\big(\sum_{i=1}^{n}I(A_i=1)I(\Delta Y_{i}=k)\big)/\sum_{i=1}^{n}I(A_i=1)}{\big(\sum_{i=1}^{n}I(A_i=0)I(\Delta Y_{i}=k)\big)/\sum_{i=1}^{n}I(A_i=0)}
$$

Using the log-normal approximation of the empirical means \citep{katz_obtaining_1978} the two incidence ratios can be approximated with 
\begin{align*}
    &&\mathrm{log}(\widehat{IR}_1) \sim \mathcal{N}(\mathrm{log}(IR_1), \sigma^2_{IR_1}) &&\mathrm{log}(\widehat{IR}_2) \sim \mathcal{N}(\mathrm{log}(IR_2), \sigma^2_{IR_2}),
\end{align*}
with estimated variance $\widehat{\sigma^2}_{IR_k}= \frac{1-\sum_{i=1}^n I(A_i=1)I(\Delta Y_{i}=k)/\sum_{i=1}^n I(A=1)}{\sum_{i=1}^n I(A_i=1)I(\Delta Y_{i}=k)}+\frac{1-\sum_{i=1}^n I(A_i=0)I(\Delta Y_{i}=k)/\sum_{i=1}^n I(A=0)}{\sum_{i=1}^n I(A_i=0)I(\Delta Y_{i}=k)},$
based on the $\delta$-method \citep{wasserman_introduction_2006}, and using the fact that the outcomes in the two arms are independent of each other (due to the assumption of no interference).

However, if we wish to derive the variance for $\widehat{IR}_1/\widehat{IR}_2$, or by the log-normal approximation, the variance of $\mathrm{log}(\widehat{IR}_1)-\mathrm{log}(\widehat{IR}_2)$, the outcomes of interest across time are no longer independent, thus the estimator
\begin{align*}
    &\widehat{\sigma^2}^*_{IR_1/IR_2}:= \widehat{\sigma^2}_{IR_1}+\widehat{\sigma^2}_{IR_2}\\
    &=\frac{1-\sum_{i=1}^n I(A_i=1)I(\Delta Y_{i}=1)/\sum_{i=1}^n I(A=1)}{\sum_{i=1}^n I(A_i=1)I(\Delta Y_{i}=1)}+\frac{1-\sum_{i=1}^n I(A_i=0)I(\Delta Y_{i}=1)/\sum_{i=1}^n I(A=0)}{\sum_{i=1}^n I(A_i=0)I(\Delta Y_{i}=1)}\\
    &+\frac{1-\sum_{i=1}^n I(A_i=1)I(\Delta Y_{i}=2)/\sum_{i=1}^n I(A=1)}{\sum_{i=1}^n I(A_i=1)I(\Delta Y_{i}=2)}+\frac{1-\sum_{i=1}^n I(A_i=0)I(\Delta Y_{i}=2)/\sum_{i=1}^n I(A=0)}{\sum_{i=1}^n I(A_i=0)I(\Delta Y_{i}=2)}
\end{align*}
for the variance of the difference of the log incidence ratios is biased.

Since the correlation between the incidence ratios is unknown, which we could derive only when we have access to the distribution of the principal strata and the distribution of the exposures over the two time-points, the estimated variance of the log-difference can be upper-bounded by
$$
    \widehat{\sigma^2}_{IR_1/IR_2}^{\mathrm{wide}} = \widehat{\sigma^2}_{IR_1}+\widehat{\sigma^2}_{IR_2}+2 \sqrt{\widehat{\sigma^2}_{IR_1}}\sqrt{\widehat{\sigma^2}_{IR_2}}.
$$

Alternatively, instead of applying the $\delta$-method separately to $IR_1$ and $IR_2$, we can directly apply it to the ratio $IR_1/IR_2$.

Denote the total number no events, events at time 1, events at time 2 with $\mathbf{M}^a=\big(M_0^a, M_1^a, M_2^a \big)$ under treatment $a$, and the total number of individuals randomized to treatment $a$ with $N^a$. Then $M^a$ has a multinomial distribution 
$$
    \mathbf{M}^a \sim \mathrm{Multinomial}\big(N^a, (p_0^a, p_1^a, p_2^a)\big),
$$
where $p_k^a$-s are determined by the distribution of the principal strata and the exposures.

The estimator of interest is
\begin{align*}
    \widehat{{IR_1}/{IR_2}}:=\left.\frac{\sum_{i=1}^nI(A_i=1)I(\Delta Y_i =1)}{\sum_{i=1}^nI(A_i=0)I(\Delta Y_i =1)} \middle/ \frac{\sum_{i=1}^nI(A_i=1)I(\Delta Y_i =2)}{\sum_{i=1}^nI(A_i=0)I(\Delta Y_i =2)} \  \right.
\end{align*}
Using the fact that $$\mathbf{M}^a =\Big(\sum_{i=1}^nI(A_i=a)I(\Delta Y_i =0), \sum_{i=1}^nI(A_i=a)I(\Delta Y_i =1), \sum_{i=1}^nI(A_i=a)I(\Delta Y_i =2)\Big)$$
and that two arms are uncorrelated, the mean and the covariance of the estimators of interest, 
\begin{align*}
    \hspace{-3em}&\mathbf{C}\\
    \hspace{-30em}&= \Big(\sum_{i=1}^nI(A_i=1)I(\Delta Y_i =1), \sum_{i=1}^nI(A_i=1)I(\Delta Y_i =2), \sum_{i=1}^nI(A_i=0)I(\Delta Y_i =1), \sum_{i=1}^nI(A_i=0)I(\Delta Y_i =2) \Big)^\top
\end{align*}
are
\begin{align*}
    &&\mathbf{\mu}:=E[\mathbf{C}] = \begin{pmatrix}
        N_1 p_1^1\\
        N_1 p_2^1\\
        N_0 p_1^0\\
        N_0 p_2^0
    \end{pmatrix}
    && \Sigma(\mathbf{C}) = \begin{pmatrix}
        N_1 p_1^1(1-p_1^1) & -N_1 p_1^1 p_2^1 &0 &0\\
        -N_1 p_1^1 p_2^1 & N_1 p_2^1(1-p_2^1) &0 &0\\
        0& 0&  N_0 p_1^0(1-p_1^0) & -N_0 p_1^0 p_2^0\\
        0 & 0 & -N_0 p_1^0 p_2^0 & N_0 p_2^0(1-p_2^0)
    \end{pmatrix}
\end{align*}
Using the multivariate $\delta$-method for some $g(U_1, U_2 ,U_3, U_4 )= \mathrm{log}(U_1)- \mathrm{log}(U_2)-\mathrm{log}(U_3)+\mathrm{log}(U_4)$ with the gradient $\nabla g(\mathbf{U}) =  \Big(\frac{1}{U_1}, -\frac{1}{U_2}, -\frac{1}{U_3}, \frac{1}{U_4}\Big)^\top$, we have that the variance of $g(\mathbf{C})$ is equal to
\begin{align*}
    \mathrm{Var}(g(\mathbf{C}))=\nabla g(\mathbf{U})^\top\big|_{\mathbf{U}=\mathbf{\mu}} \Sigma(\mathbf{C}) \nabla g(\mathbf{U})\big|_{\mathbf{U}=\mathbf{\mu}}= \frac{1}{N_1p_1^1}+\frac{1}{N_1p_1^2}+\frac{1}{N_0p_1^0}+\frac{1}{N_0p_1^0},
\end{align*}
which we can estimate by
\begin{align*}
\widehat{\mathrm{Var}(g(\mathbf{C}))} &= \frac{1}{\sum_{i=1}^nI(A_i=1)I(\Delta Y_i =1)}+\frac{1}{\sum_{i=1}^nI(A_i=1)I(\Delta Y_i =2)}\\
&+\frac{1}{\sum_{i=1}^nI(A_i=0)I(\Delta Y_i =1)}+\frac{1}{\sum_{i=1}^nI(A_i=0)I(\Delta Y_i =2)}
\end{align*}

The derivation of the variance shows that using incidence ratios instead of hazard ratios not only has the benefit of stability when the null hypothesis holds, but also has statistical advantages. Since the hazard ratio is calculated across different populations over time, the two arms cannot be modeled as two multinomial distributions. Their correlation is not known, as it depends on the principal strata and the exposure of individuals; therefore, to derive valid statistical tests, we have to rely on non-parametric bootstrap. Thus, having patient-level data is necessary to assess the equality of hazard ratios over time, while for the incidence ratios, summary data is sufficient.

When we only have access to summary data, such as the analysis of the BNT162b2 vaccine against COVID-19, we follow the Estimation procedure discussed by \citet{janvin_quantification_2024} in their Appendix.

\subsection{Conditional waning}\label{APP:Cond wane}
The challenge effect, and thus waning, can be defined conditional on baseline covariates.
\begin{definition}[Conditional challenge effect]
Denote the measured baseline covariates $\mathbf{L} \in \mathbf{\mathcal{L}}$. Then the conditional challenge effects at times 1 and 2 are
\begin{align*}
        VE_1^{\mathrm{challenge}}(\mathbf{l}) &= 1- \frac{E[\Delta Y_1^{a=1, e_1=1}|\mathbf{L}=\mathbf{l}]}{E[\Delta Y_1^{a=0, e_1=1}|\mathbf{L}=\mathbf{l}]},\\
        VE_2^{\mathrm{challenge}}(\mathbf{l}) &= 1- \frac{E[\Delta Y_1^{a=1, e_1=0, e_2=1}|\mathbf{L}=\mathbf{l}]}{E[\Delta Y_1^{a=0, e_1=0, e_2=1}|\mathbf{L}=\mathbf{l}]}
    \end{align*}
\end{definition}
Correspondingly, the incidence rates are
\begin{align*}
        &&IR_1(\mathbf{l}):= \frac{E[\Delta Y_1|A=1, \mathbf{L}=\mathbf{l}]}{E[\Delta Y_1|A=0, \mathbf{L}=\mathbf{l}]}  &&IR_2(\mathbf{l}):=\frac{E[\Delta Y_2|A=1, \mathbf{L}=\mathbf{l}]}{E[\Delta Y_2|A=0, \mathbf{L}=\mathbf{l}]}
    \end{align*}

Suppose that Assumption \ref{ASS: blind exp} holds only conditionally on the baseline covariates.
\begin{assumption}[Conditional exposure exchangeability]\label{ASS: cond exp exch}
    \begin{align*}
        E_1^{a} &\independent \Delta Y_2^{a, e_1=0}|E_2^{a, e_1=0}, \mathbf{L}, &E_1^{a}&\independent \Delta Y_1^{a, e_1}| \mathbf{L},\\
        E_2^{a, e_1=0} &\independent \Delta Y_2^{a, e_1=0, e_2}|\mathbf{L}.
    \end{align*}
\end{assumption}

For example, suppose that there are some \textit{measured} confounders between the exposures and the outcome, such as age or physical fitness. Then, in order for the independence to hold between the outcomes and the exposures, all confounders should be conditioned on.

Similarly, if randomization occurs conditional on the baseline covariates, Assumption \ref{ASS: Trt exch} and \ref{ASS: pos} hold conditionally.

\begin{assumption}[Conditional treatment exchangeability]\label{ASS: cond trt exch}
    $$
        E_1^{a}, \Delta Y_1^{a}, \Delta Y_1^{a, e_1}, E_2^{a, e_1}, E_2^{a, e_1 \delta y_1=0}, \Delta Y_2^{a, e_1,  \delta y_1=0, e_2} \independent A | \mathbf{L}
    $$
    for all $a, e_1, e_2 \in \{0,1\}$.
\end{assumption}

\begin{assumption}[Conditional positivity]\label{ASS: cond pos}
    $$P(A=a| \mathbf{L}=\mathbf{l})>0$$
    
    $\text{ for all } a \in \{0,1\}, \text{ and for all } \mathbf{l} \in \mathbf{\mathcal{L}}, \text{ such that } P(\mathbf{L}=\mathbf{l})>0.$
\end{assumption}

Analogously to the proof in Appendix \ref{APP: Prop proof} we can show that under Hypothesis \ref{HYP: formal null} and Assumptions \ref{ASS: Exp nec}-\ref{ASS: blind exp} and \ref{ASS: cond exp exch}
$$
    IR_1(\mathbf{l}) = IR_2(\mathbf{l}) \quad \forall  \, \mathbf{l} \in \mathcal{L}
$$

Then, the null hypothesis of no waning may be evaluated separately across subgroups; however, the analyses should control for multiple testing \citep{simes_improved_1986, benjamini_controlling_1995}.

Alternatively, assume constant proportionality across the subgroups defined by $\mathbf{L}$.
\begin{assumption}[Constant incidence ratios]\label{ASS: Const IR}
    \begin{align*}
         E[\Delta Y_1|A=1, \mathbf{L}=\mathbf{l}]&= \alpha_1E[\Delta Y_1|A=0, \mathbf{L}=\mathbf{l}]\\
         E[\Delta Y_2|A=1, \mathbf{L}=\mathbf{l}]&=\alpha_2 E[\Delta Y_2|A=0, \mathbf{L}=\mathbf{l}]\\
    \end{align*}
    for all $\mathbf{l} \in \mathcal{L}$.
\end{assumption}
Under the null hypothesis and Assumptions \ref{ASS: Exp nec} and \ref{ASS: cond exp exch}, we have that  $IR_1(\mathbf{l}) =IR_2(\mathbf{l}) $. Thus, it follows that in this setting $\alpha_1 = \alpha_2$ must hold. Then
\begin{align*}
   \frac{IR_1}{IR_2}= \frac{\int IR_1(\mathbf{l}) \mathrm{d}P(\mathbf{L})}{\int IR_2(\mathbf{l}) \mathrm{d}P(\mathbf{L})} = \int \frac{ IR_1(\mathbf{l})}{ IR_2(\mathbf{l})} \mathrm{d}P(\mathbf{L}) = \frac{\alpha_1}{\alpha_2} \int \frac{1}{1} \mathrm{d}P(\mathbf{L}) = 1,
\end{align*}
therefore, the null hypothesis implies that there is no marginal population-level waning.

\subsection{Testing for multivalued exposure}\label{APP: multi exp}
Assumption \ref{ASS: Cons} relies on the binary categorization of the exposure. Certain individuals are unexposed, and all others are exposed according to some fixed intensity of exposure. Then the challenge effects, with respect to which waning is defined, are formulated in terms of this unspecified fixed exposure. However, in certain scenarios, defining exposure as binary, is implausible, for example, if multiple strains circulate in the population with different infectivity, the intensity of exposure varies according to the strains, or, when follow-up period spans over longer periods, certain individuals might face repeated exposures, which leads to a definition of a multi-valued exposure.

We revisit the findings of \citet{janvin_quantification_2024} in their Appendix D, which shows that under a stability assumption on the distribution of the ``strength of exposure", conditional on being exposed, we can still test for the waning of the vaccine, using the partial-identification bounds derived. However, by not including in the null hypothesis the restriction of binary intensity, considerable power is lost.

Adopting the setup and the notation of \citet{janvin_quantification_2024}, suppose that there exists a finer-scale version of the randomized trial described in the main text, denoted with $G$. In this modified challenge trial $G$, individuals are randomized to treatment or control ($A \in \{0,1\}$) at baseline as before, and then at time 1 and time 2, to some intensity of exposure, denoted with $Q_1, Q_2 \in \mathcal{Q}$, respectively. $\mathcal{Q}$ can be interpreted as the conceivable range of the number of infectious particles per exposure, the number of exposures during the follow-up period, or the cross-product of the former two. Suppose that no exposure is well-defined for $Q_1$ and $Q_2$ as well, that is, there are no multiple versions of no exposure in the modified trial. In $G$, the coarse-grained exposure is defined as
$$
E_k=\mathbb{I}(Q_k\neq 0) \; \text{ and }\; E_2^{q_1=0}= \mathbb{I}(Q_2^{q_1=0}\neq 0).
$$

For brevity, we do not restate all the assumptions of \citet{janvin_quantification_2024}, interested readers should see their Assumptions 9-12 in their Appendix D. However, we restate Proposition 2, which is central to our argument:

\begin{proposition}[Restatement of Proposition 2 in \citet{janvin_quantification_2024}]\label{PROP: janvin}
Under Assumptions 9-12 of \citet{janvin_quantification_2024} and Assumptions \ref{ASS: Cons} (i) and \ref{ASS: pos} of this manuscript,
\begin{align*}
    E[\Delta Y_1^{a}|E_1^{a}=1, A=a]&=E_G[E_G[\Delta Y_1^{a, Q_1}]|E_1=1, A=a]\\
    E[\Delta Y_2^{a, e_1=0}|E_2^{a, e_1=0}=1, A=a]&=E_G[E_G[\Delta Y_2^{a, q_1=0, Q_2^{q_1=0}}]|E_2^{q_1=0}=1, A=a].
\end{align*}

    Under the further Assumptions \ref{ASS: Exp nec}, \ref{ASS: blind exp}, \ref{ASS: Trt exch} from our manuscript and Assumption 4 of \citet{janvin_quantification_2024},

    \begin{align}
        1-\frac{E_G[E_G[\Delta Y_1^{a=1, Q_1}]|E_1=1, A=1]}{E_G[E_G[\Delta Y_1^{a=0, Q_1}]|E_1=1, A=0]}]&=1-\frac{E[\Delta Y_1|A=1]}{E[\Delta Y_1|A=0]}\label{EQ: multi exp 1}\\
        \hspace{-1em}1-\frac{E_G[E_G[\Delta Y_1^{a=1, q_1=0, Q_2^{q_1=0}}]|E_2^{q_1=0}=1, A=1]}{E_G[E_G[\Delta Y_1^{a=0, q_1=0, Q_2^{q_1=0}}]|E_2^{q_1=0}=1, A=0]}]&\in \left[1-\frac{E[\Delta Y_1+\Delta Y_2|A=1]}{E[\Delta Y_2|A=0]}, 1-\frac{E[\Delta Y_2|A=1]}{E[\Delta Y_1+\Delta Y_2|A=0]}\right]\label{EQ: multi exp 2}
    \end{align}
\end{proposition} 

See proof of the proposition in Appendix D of \citet{janvin_quantification_2024}.

\citet{janvin_quantification_2024} provided an interpretation of equations \eqref{EQ: multi exp 1}-\eqref{EQ: multi exp 2}: the right hand sides are identification formulas for the ratio of expected potential outcomes, under exposure intervention in which $e_k=0$ corresponds to isolation, and $e_k=1$ corresponds to an exposure with randomized intensity, where the investigator draws the intensity $Q_1$ and $Q_2^{q_1=0}$ from the distribution $F_{Q_1|E_1=1, A=a}$ and $F_{Q_2^{q_1=0}|E_2^{q_1=0}=1, A=a}$.

Even if there is no population-level waning, that is, if
\begin{align*}
    E_G[\Delta Y_1^{a=1, q_1=q}] &= E_G[\Delta Y_2^{a=1, q_1=0, q_2=q}]\\
    E_G[\Delta Y_1^{a=0, q_1=q}]&=E_G[\Delta Y_2^{a=0, q_1=0, q_2=q}]
\end{align*}
  with probability 1, $\forall q \in \mathcal{Q}$,
holds, the left hand sides of Equations \eqref{EQ: multi exp 1}-\eqref{EQ: multi exp 2} are not necessarily equal. If the distribution of intensities changes from time 1 to time 2, for example, because a new strain enters the population during follow-up, with a different infectivity, then the expectations taken with respect to the conditional distributions $F_{Q_1|E_1=1, A=a}$ and $F_{Q_2^{q_1=0}|E_2^{q_1=0}=1, A=a}$ need not coincide.

Thus, assume the following, stated first in \citet{janvin_quantification_2024}:

\begin{assumption}[Stationarity of exposure versions among the exposed]\label{ASS: Stat exp}
    $$P_G(Q_1\leq q|E_1=1, A=a)= P_G(Q_2^{q_1=0}\leq q|E_2^{q_1=0}, A=a) \; w.p. \; 1.$$
\end{assumption}

Assumption \ref{ASS: Stat exp} allows for different distribution of exposures ($E_1$, $E_2$) between the intervals, however, it requires that, conditional on being exposed in a given interval, the distribution of the ``strength of the exposure" is the same in the two.

With a slight reformulation of Proposition 3 of \citet{janvin_quantification_2024} we can state the following:
\begin{proposition}[Restatement of Proposition 3 in \citet{janvin_quantification_2024}]\label{PROP: janvin 2}
    Under Assumptions \ref{ASS: Exp nec},\ref{ASS: blind exp}, \ref{ASS: Cons} (i), (ii), \ref{ASS: Trt exch}, \ref{ASS: pos}, and \ref{ASS: Stat exp} from our manuscript and Assumptions 4, 9-12 from \citet{janvin_quantification_2024}, the null hypothesis
    \begin{align*}
    E_G[\Delta Y_1^{a=1, q_1=q}] &= E_G[\Delta Y_2^{a=1, q_1=0, q_2=q}]\\
    E_G[\Delta Y_1^{a=0, q_1=q}]&=E_G[\Delta Y_2^{a=0, q_1=0, q_2=q}]
\end{align*}
for all $q \in \mathcal{Q}$ implies
\begin{equation}\label{EQ: multiwaning test}
    \frac{E[\Delta Y_2|A=1]}{E[\Delta Y_1+\Delta Y_2|A=0]}\leq \frac{E[\Delta Y_1|A=1]}{E[\Delta Y_1|A=0]}\leq \frac{E[\Delta Y_1+\Delta Y_2|A=1]}{E[\Delta Y_2|A=0]}
\end{equation}
\end{proposition}
See proof in Appendix D of \citet{janvin_quantification_2024}.

Thus, two one-sided tests based on observed quantities can be constructed for the null hypothesis of no waning, whether the confidence interval for the difference between $\frac{E[\Delta 
Y_1|A=1]}{E[\Delta 
Y_1|A=0]}$ and the corresponding bounds excludes $0$. However, the rejection of the null hypothesis implies only that there exists at least one intensity $q\in \mathcal{Q}$ for which the expected potential outcome changed over time, but we cannot identify, the subset $\mathcal{Q'}\subseteq \mathcal{Q}$, for which the vaccine has waned.

Our null hypothesis, Hypothesis \ref{HYP: formal null} is contingent on binary exposure intensity: an individual is either exposed or not. While waning under multiple versions of exposures can still be tested, possibly wide bounds
$$
\left[\frac{E[\Delta Y_2|A=1]}{E[\Delta Y_1+\Delta Y_2|A=0]},\frac{E[\Delta Y_1+\Delta Y_2|A=1]}{E[\Delta Y_2|A=0]}\right]
$$
come at the cost of substantial loss of power. Therefore, the originally proposed test can be interpreted as a joint hypothesis of no waning and an all or nothing vaccine effect under exposure. When the original test of the equality of incidence ratios is rejected, this may reflect either waning vaccine effects or ``leaky" vaccines \citep{halloran_thirty-five_2024}, in which framework, the expected potential outcome under exposure is sensitive to exposure intensity.

The reason why, in settings with multiple exposure intensities, not even a sharp null hypothesis of no waning yields a valid test based on the equality of incidence ratios is that the joint distribution of exposure intensities ($Q_1, Q_2$) and susceptibility ($U_Y$) is unknown. It is possible that some individuals exposed at time 1 had underlying susceptibility such that $\Delta Y_1^{a, q_1=q}=1$ for any $q>0$. Had they been isolated and then exposed at time 2, they would develop the outcome regardless of exposure intensity, even if no waning occurred. By contrast, among those who developed the outcome at time 1, only because the exposure intensity exceeded some threshold $q^*$, had they been isolated in the first interval, the expected outcome under exposure at time 2 can be expressed as a function of the intensity distributions and $q^*$, using the stationarity Assumption \ref{ASS: Stat exp}. We conjecture that by explicitly modelling the joint distribution of susceptibility and exposure intensity, under the assumptions of Proposition \ref{PROP: janvin 2}, tighter bounds can be derived for the time 2 challenge effect, thereby increasing the power of our test. Building on results on vaccine effects in the presence of competing strains \citep{perenyi_variant_2025}, this direction will be the focus of follow-up work.

\subsection{Sensitivity analysis of the depletion of susceptibles}\label{APP: heterogen}

In this section, we present a sensitivity analysis that uses investigators' beliefs about the extent of depletion of susceptible individuals to examine the plausibility of the null hypothesis of no waning. 
Following the main text, we present our results in a setting with $2$ intervals; however, extension to a setting with $K$ intervals follows straightforwardly. 
As in \citet{janvin_quantification_2024}, we define the parameter $\psi$ as the ratio of the expected potential outcomes under treatment at the two timepoints, in a challenge trial:
$$
\psi := \frac{E[\Delta Y_1^{a=1, e_1=1}]}{E[\Delta Y_2^{a=1, e_1=0, e_2=1}]}.
$$
Suppose that no waning of the placebo holds, at least in the population sense, that is 
\begin{align}
    E[\Delta Y_1^{a=0, e_1=1}]=E[\Delta Y_2^{a=0, e_1=0, e_2=1}], \label{eq:no_waning_placebo_expectation}
\end{align}
then 
$$
\psi = \frac{1-VE_1^{\mathrm{challenge}}}{1-VE_2^{\mathrm{challenge}}}.
$$
Observing $\psi\neq 1$ shows that the challenge effect changed over time, thus the vaccine has waned in either direction. \citet{janvin_quantification_2024} describe the challenge trial corresponding to the estimand $\psi$, which bears similarities to the trials using delayed vaccination to assess waning, when the infectious agent enters the population only after a certain calendar time \citep{lipsitch_depletion--susceptibles_2019, ray_depletion--susceptibles_2020}.

The following proposition decomposes the challenge effect into a ratio of hazard ratios (HR ratio) and an error term that reflects the departure of the challenge effect from the HR ratio.
\begin{proposition}\label{prp:hazard_ratio_decomposition}
    Under Assumptions \ref{ASS: Exp nec}, \ref{ASS: blind exp}, \ref{ASS: exp exch}, \ref{ASS: Cons}, \ref{ASS: Trt exch}, \ref{ASS: pos}, and \eqref{eq:no_waning_placebo_expectation}
    \begin{align}
    &\psi= \frac{E[\Delta Y_1\mid A=1]}{E[\Delta Y_1\mid A=0] }\times \frac{E[\Delta Y_2\mid \Delta Y_{1}=0, A=0]}{E[\Delta Y_2\mid \Delta Y_{1}=0, A=1] } \notag\\
    &\qquad\times \frac{E[\Delta Y_2^{e_1=0}\mid A=0]/E[\Delta Y_2\mid \Delta Y_{1}=0,A=0]  }{E[\Delta Y_2^{e_1=0}\mid A=1]/E[\Delta Y_2\mid \Delta Y_{1}=0,A=1]} \label{eq:hazard_ratio_decomposition}.
\end{align}
\end{proposition}
The proof of Proposition~\ref{prp:hazard_ratio_decomposition} follows from the proof of Theorem~1 in \citet{janvin_quantification_2024}. The error term in the second line of \eqref{eq:hazard_ratio_decomposition} can differ from 1 due to $(i)$ differential depletion of susceptible individuals over time and $(ii)$ differential acquisition of natural immunity from undetected infections in interval 1 across treatment arms. In the absence of $(i)$ and $(ii)$, the challenge effect is equal to the HR ratio, as formalized in Proposition~1 of \citet{janvin_quantification_2024}.
When data is available on infectious outcomes not targeted by the treatment, that is, negative control outcomes, the bias introduced by depletion of susceptibles can be assessed using methods introduced by \citet{ashby_debiasing_2025}. Consider the two-interval setting introduced in the main text, where in addition to the primary outcome, $\Delta N_k$ indicates whether the negative control outcome occurs during interval $k$. Then, adapting the estimand introduced by \citet{ashby_debiasing_2025}, the relative change in susceptibility between the two arms can be measured by
$$
\frac{E[\Delta N_2|\Delta Y_1 = \Delta N_1=0, A=1]}{E[\Delta N_2|\Delta Y_1 = \Delta N_1=0, A=0]}.
$$
Under suitable assumptions on the confounding structure, this quantity can be used to debias the hazard ratio, thereby improving our ability to detect waning. Even though for similar outcomes such as COVID-19 infection and other acute respiratory infections, it may be reasonable to assume that the same unmeasured variables act as confounders for both outcomes, the functional form of their effects may differ across outcomes. Thus, while exact results on identification require strong and untestable assumptions, the approach of \citet{ashby_debiasing_2025} may still be useful as an exploratory tool for detecting residual confounding.

Next, we will describe a sensitivity analysis that uses an investigator's belief about the size of the error term in \eqref{eq:hazard_ratio_decomposition}. Suppose an investigator believes that 
\begin{align}
    1\leq \frac{E[\Delta Y_2^{e_{1}=0}\mid A=a]}{E[\Delta Y_2\mid \Delta Y_{1}=0,A=a]} \leq \alpha \label{eq:heterogeneity_belief}
\end{align}
for all $a\in\{0,1\}$. In words, the investigator believes that isolation during interval 1, which retains susceptible individuals in the risk set and prevents natural immunity from undetected infections, can at most lead to an $\alpha$-fold increase in infections during interval 2 in either treatment group. The lower bound in \eqref{eq:heterogeneity_belief} is a modest tightening of the bound in Assumption 4 (exposure effect restriction) of \citet{janvin_quantification_2024} reflecting that isolation does not make the population at risk more protected against infectious exposures in the future.

\begin{proposition}[Sensitivity analysis\label{prp:sensitivity analysis}]
    Suppose Assumptions \ref{ASS: Exp nec}, \ref{ASS: blind exp}, \ref{ASS: exp exch}, \ref{ASS: Cons}, \ref{ASS: Trt exch}, \ref{ASS: pos} and \eqref{eq:no_waning_placebo_expectation}, \eqref{eq:heterogeneity_belief} hold. Next, suppose that $\psi=1$, i.e.\ the null hypothesis of no waning holds in expectation. Then
    \begin{align}
        \bigg(\frac{E[\Delta Y_1\mid A=1]}{E[\Delta Y_1\mid A=0] }\times \frac{E[\Delta Y_2\mid \Delta Y_{1}=0, A=0]}{E[\Delta Y_2\mid \Delta Y_{1}=0, A=1] }\bigg)^{-1} \leq  \alpha \label{eq:smallest_heterogeneity}~.
    \end{align}
\end{proposition}
Expression \eqref{eq:smallest_heterogeneity} follows directly from using \eqref{eq:heterogeneity_belief} and $\psi=1$ in \eqref{eq:hazard_ratio_decomposition}. The left hand side of \eqref{eq:smallest_heterogeneity} is the (reciprocal) HR ratio, which is only a functional of the observed data, whereas $\alpha$ is the investigator's belief about the maximal possible value of the ratios $E[\Delta Y_2^{e_{1}=0}\mid A=a]/E[\Delta Y_2\mid \Delta Y_{1}=0,A=a]$ for $a\in\{0,1\}$. If the investigator's chosen value of $\alpha$ is too small to satisfy \eqref{eq:smallest_heterogeneity}, we obtain a contradiction: the observed HR ratio cannot alone be explained by the investigators' belief ($\alpha$) about the depletion of susceptible individuals. For example, suppose that
\begin{align*}
    &\alpha =5, \\
    &E[\Delta Y_1\mid A=1]=0.01, \quad E[\Delta Y_1\mid A=0]=0.1,\\
    &E[\Delta Y_2\mid \Delta Y_1=0,A=1]=E[\Delta Y_2\mid \Delta Y_1=0,A=0]=0.1,
\end{align*}
then Proposition~\ref{prp:sensitivity analysis} gives the contradiction $10\leq 5$. Whenever \eqref{eq:smallest_heterogeneity} leads to a contradiction, this implies that the premises assumed in Proposition~\ref{prp:sensitivity analysis} are false: in particular, the null hypothesis of no waning ($\psi=1$) may be implausible in light of the investigators' belief $(\alpha)$. Although it is not obvious how investigators can agree on a particular choice of $\alpha$, Proposition~\ref{prp:sensitivity analysis} offers a way to quantitatively state beliefs about the depletion of susceptible individuals and use these to assess the plausibility of the null hypothesis of no waning.

\subsection{Interpretation of the p-value under parametric assumptions}\label{APP: p-val}
We introduced a test for waning that is valid under non-parametric assumptions and is agnostic to the exact form of the data-generating mechanism. Under these assumptions, while the p-value is informative about the existence of waning, it is not necessarily informative about its magnitude, that is, a lower p-value does not necessarily imply a larger discrepancy between the time 1 and time 2 challenge effects. However, if the mechanism underlying waning is specified, for example, by defining waning as a particular transition from one principal stratum to another, then stronger claims can be made about the magnitude of waning.

Similarly to Appendix \ref{APP: Simulation} consider the following data-generating mechanism: 
\begin{align*}
    U_Y& \sim F_{U_Y}\\
    U_E& \sim F_{U_E}\\
    A &:= \mathrm{Bernoulli}(\pi)\\
    T_1|U_Y &:=\mathrm{Categorical}(p_{doomed_1}(U_Y), p_{helped_1}(U_Y), p_{harmed_1}(U_Y), p_{immune_1}(U_Y))\\
    W| U_Y,T_1
& :=
\begin{cases}
\mathrm{Bernoulli}(\omega(U_Y)), & T_1=\mathrm{helped},\\
0, & T_1\neq \mathrm{helped}.
\end{cases}\\
    T_2|T_1, W &:=
\begin{cases}
\mathrm{doomed}, & T_1=\mathrm{doomed},\\
\mathrm{doomed}, & T_1=\mathrm{helped},\ W=1,\\
\mathrm{helped}, & T_1=\mathrm{helped},\ W=0,\\
\mathrm{harmed}, & T_1=\mathrm{harmed},\\
\mathrm{immune}, & T_1=\mathrm{immune}.
\end{cases}\\
    E_1|U_E &:= \mathrm{Bernoulli}(p_{E_1}(U_E))\\
    \Delta Y_1 &:= E_1 \times \big(I(T_1=\text{doomed})+(1-A)\times I(T_1=\text{helped})+A\times I(T_1=\text{harmed})\big)\\
    E_2|U_E, \Delta Y_1 &:= \mathrm{Bernoulli}(p_{E_2}(U_E, \Delta Y_1))\\
    \Delta Y_2 &:= (1-\Delta Y_1) \times E_2 \times \big(I(T_2=\text{doomed})+(1-A)\times I(T_2=\text{helped})+A\times I(T_2=\text{harmed})\big).
\end{align*}
In the proposed data-generating mechanism, we assumed that the only possible transition between principal strata over time is from \textit{helped} to \textit{doomed}. This transition corresponds to a decrease in vaccine efficacy, while the placebo response is assumed not to wane. Therefore, individual-level vaccine improvement is ruled out by construction, which is a strong and possibly unrealistic assumption. This can be interpreted as a monotonicity assumption over time, requiring that susceptibility under treatment can only increase, not decrease.

For ease of notation, let $S_{a,k}$ denote the indicator that an individual would develop the outcome at time $k$ under treatment $a$ and exposure, had they remained event-free up to that time. For example, $S_{1,1}=1$ indicates that the individual belongs to $T_1 \in \{\text{doomed}, \text{harmed}\}$, whereas $S_{0,2}=1$ indicates that the individual belongs to $T_2 \in \{\text{doomed}, \text{helped}\}$. In this notation, we have
\begin{align*}
    1-VE_1^{\mathrm{challenge}}&=\frac{P(S_{1,1}=1)}{P(S_{0,1}=1)}=\frac{P(T_1 \in \{\text{doomed, harmed}\})}{P(T_1 \in \{\text{doomed, helped}\})},\\
    1-VE_2^{\mathrm{challenge}}&=\frac{P(S_{1,2}=1)}{P(S_{0,2}=1)}=\frac{P(T_2 \in \{\text{doomed, harmed}\})}{P(T_2 \in \{\text{doomed, helped}\})}.
\end{align*}
In contrast, the incidence ratios are
\begin{align*}
    IR_1&=\frac{P(E_1=1, S_{1,1}=1)}{P(E_1=1, S_{0,1}=1)},\\
    IR_2&=\frac{P(E_2=1, \Delta Y_1=0,  S_{1,2}=1)}{P(E_2=1, \Delta Y_1=0, S_{0,2}=1)}.
\end{align*}
Since in the specified data-generating mechanism no waning of the placebo holds, $S_{0,1}=S_{0,2}$. Moreover, because the only possible transition between principal strata is from \textit{helped} to \textit{doomed}, we have $S_{1,1}\leq S_{1,2}$

Since the principal strata, $T_1$, is independent of the exposure $E_1$, similarly to Proposition \ref{PROP: inv incidence}, we obtain
$$
IR_1=\frac{P(E_1=1, S_{1,1}=1)}{P(E_1=1, S_{0,1}=1)}=\frac{P(S_{1,1}=1)}{P(S_{0,1}=1)}=1-VE_1^\mathrm{challenge}.
$$
To rewrite $IR_2$, we consider the two arms separately. Using the notation for susceptibility, we can rewrite 
$$
E[\Delta Y_2] = \begin{cases}
    E[(1-\Delta Y_1)E_2 S_{1,2}]=E[(1-E_1S_{1,1}) E_2 S_{1,2}] \text{ if } A=1,\\
    E[(1-\Delta Y_1)E_2 S_{0,2}]=E[(1-E_1S_{0,1}) E_2 S_{0,2}] \text{ if } A=0.
\end{cases}
$$
Since no waning of the placebo holds, $(1-E_1S_{0,1})S_{0,2}=(1-E_1)S_{0,1}$, thus
$$
E[\Delta Y_2|A=0]=E[(1-E_1)p_{E_2}(0, U_E)]P(S_{0,1}=1),
$$
where we further used that we are only interested in the time 2 exposure of the survivors, thus making the time 2 exposure independent of the time 1 susceptibility.

For the treated, $S_{1,2}=1$ can occur in two disjoint ways: either the individual was already susceptible at time 1, $S_{1,1}=1$, or the individual becomes newly susceptible due to waning, $\{S_{1,1}=0, W=1, S_{1,2}\}=\{T_1=\text{helped}, W=1\}$. The first group would develop the outcome at time 1, if they were exposed, hence, they contribute to the time 2 events only if they were unexposed at time 1 and then exposed at time 2, which happens with probability $E[(1-E_1)p_{E_2}(0, U_E)]$.

In the second group, individuals remain event-free at time 1, regardless of exposure, thus only the time 2 exposure needs to be considered, which occurs with probability $E[p_{E_2}(0, U_E)]$. Putting these two groups together, the observed time 2 event rate under treatment is 
\begin{align*}
    E[\Delta Y_2|A=1]&=E[(1-E_1)p_{E_2}(0, U_E)]P(S_{1,1}=1)\\
    &+E[p_{E_2}(0, U_E)]P(S_{1,1}=0, W=1, S_{1,2}=1),
\end{align*}
thus $IR_2$ is
\begin{align*}
    IR_2&=\frac{E[(1-E_1)p_{E_2}(0, U_E)]P(S_{1,1}=1)
    +E[p_{E_2}(0, U_E)]P(S_{1,1}=0, W=1, S_{1,2}=1)}{E[(1-E_1)p_{E_2}(0, U_E)]P(S_{0,1}=1)}.
\end{align*}

Our test statistic is $\mathrm{log}(\widehat{IR}_2)-\mathrm{log}(\widehat{IR}_1)=\mathrm{log}\left(\frac{\widehat{IR}_2}{\widehat{IR}_1}\right)$. Under the proposed data-generating mechanism, the ratio of the incidence ratios is equal to
\begin{align*}
    \frac{IR_2}{IR_1}&=\frac{E[(1-E_1)p_{E_2}(0, U_E)]P(S_{1,1}=1)
    +E[p_{E_2}(0, U_E)]P(S_{1,1}=0, W=1, S_{1,2}=1)}{E[(1-E_1)p_{E_2}(0, U_E)]P(S_{1,1}=1)}\\
    &=1+\frac{E[p_{E_2}(0, U_E)]}{E[(1-E_1)p_{E_2}(0, U_E)]}\frac{P(S_{1,1}=0, W=1, S_{1,2}=1)}{P(S_{1,1}=1)}.
\end{align*}
However, $P(S_{1,1}=0, W=1, S_{1,2}=1)$ is precisely the proportion of individuals for whom waning has occurred, that is, the extent of waning. By analogous rewriting, we can show that
$$
\frac{1-VE_2^\mathrm{challenge}}{1-VE_1^\mathrm{challenge}}=1+\frac{P(S_{1,1}=0, W=1, S_{1,2}=1)}{P(S_{0,1}=1)}.
$$
Therefore, provided that the probability of exposure among survivors at time 2 is non-zero, the test statistic is strictly increasing in the waning term. Equivalently, as the ratio  $\frac{1-VE_2^\mathrm{challenge}} {1-VE_1^\mathrm{challenge}}$ increases, the test statistic increases as well.

In conclusion, under the proposed parametrized version of waning, the test statistic and the corresponding p-value are informative about the extent of waning. Thus, under suitable modeling assumptions, the test can be used not only to detect waning but also to quantify it.

\end{document}